\newcommand{\layout}{twocolumn}
\documentclass[journal,twoside,\layout,web]{ieeecolor}

\usepackage{etoolbox} % Provides \ifstrequal
\usepackage{ifthen} % for if-then-else conditionals
\usepackage{xfp}
\usepackage{xstring}
\usepackage{adjustbox}
\newbool{isSingle}
\IfStrEq{\layout}{onecolumn}
  {\booltrue{isSingle}}
  {\boolfalse{isSingle}}

\ifbool{isSingle}
  {\newcommand{\plotfactor}{0.45}}   % if true, use 0.8
  {\newcommand{\plotfactor}{0.9}}  % if false, use 0.45

\usepackage{cite}
\usepackage{amsmath,amssymb,amsfonts}
\usepackage{graphicx}
\usepackage[hidelinks]{hyperref}
\usepackage{textcomp}
\usepackage{xcolor}

\newlength{\spEm}
\newlength{\spCm}
\newcommand{\preSection}      {\vspace{-0.85\spEm}}

\newcommand{\preSectionHard}  {\vspace{-1.7\spEm}}
\newcommand{\postFloat}       {\vspace{-0.85\spEm}}
\newcommand{\preBio}          {\vspace{-1.0\spCm}}
\newcommand{\preApdx}          {\vspace{-0.5\spCm}}
\newcommand{\algoSepIn}       {\vspace{0.25\spEm}}
\newcommand{\algoSep}         {\vspace{0.35\spEm}}

\newlength{\bioPicWidth}
\newlength{\bioPicHeight}
\newcommand{\bioPic}[1]{\includegraphics[width=\bioPicWidth,height=\bioPicHeight,clip,keepaspectratio]{#1}}

\usepackage{verbatim}
\newcommand\Set[1]{\mathbb{#1}} % override the default set command
\usepackage{booktabs}
\usepackage{bm}

\newcommand{\txt}[1]{{\textsf{#1}}}
\newcommand{\tbm}[1]{\textbf{\textsf{#1}}} %

\newcommand{\defeq}{\triangleq}

\usepackage{scalerel}

\newcommand{\xmark}{\ding{55}}
\newcommand{\raum}[1]{\mathcal{#1}}
\newcommand{\RKHS}{\ensuremath{\raum{H}}}
\newcommand{\adjoint}[1]{\ensuremath{{#1}^{*}}}
\usepackage{nicefrac}
\makeatother
\newtheorem{definition}{Definition}
\newtheorem{assumption}{Assumption}

\newtheorem{lemma}{Lemma}
\newtheorem{theorem}{Theorem}
\newtheorem{corollary}{Corollary}
\newtheorem{proposition}{Proposition}
\newtheorem{remark}{Remark}

\usepackage{tikz-cd,adjustbox}
\usepackage{footnote}
\makesavenoteenv{tabular}
\usepackage{mathrsfs}
\usepackage{pifont}

\usepackage{stackengine}

\newcommand {\abs}[1]{\left\vert\, #1 \,\right\vert}
\newcommand{\norm}[1]{\lVert #1 \rVert}

\usepackage{amssymb}
\providecommand{\norm}[1]{\lVert#1\rVert}
\providecommand{\SVDr}[1]{[\![#1]\!]_r}
\providecommand{\abs}[1]{\lvert#1\rvert}

\providecommand{\norm}[1]{\lVert#1\rVert}
\providecommand{\SVDr}[1]{[\![#1]\!]_r}
\providecommand{\abs}[1]{\lvert#1\rvert}

\newcommand{\lilsum}{\textstyle\sum}

\newcommand{\HS}[1]{{\rm{HS}}\left(#1\right)} %

\newcommand{\reg}{\gamma}

\newcommand{\EEstim}{\widehat{\Estim}}  %
\newcommand{\Estim}{{G}}  %
\newcommand{\infEstim}{{G}}
\newcommand{\EAEstim}{\widehat{\AEstim}}  %
\newcommand{\AEstim}{{A}}  %
\newcommand{\EBEstim}{\widehat{\BEstim}}  %
\newcommand{\BEstim}{{B}}  %

\newcommand{\innerprod}[2]{\langle #1 {,} #2 \rangle}

\usepackage{mathtools}

\newcommand{\scheduleVar}{\txt{\textbf{p}}}

\newcommand{\mpcB}{\txt{\textbf{B}}}

\DeclareMathAlphabet{\mathsfit}{T1}{\sfdefault}{\mddefault}{\itdefault}
\SetMathAlphabet{\mathsfit}{bold}{T1}{\sfdefault}{\mddefault}{\sldefault}
 
\newcommand{\ES}{{\widehat{S}}} %
\newcommand{\aES}{{\widehat{S}^*}} %
\newcommand{\nyES}{{\widetilde{S}}} %
\RequirePackage{xcolor}           % already loaded by many classes, but ensure it’s there

\newcommand{\concat}[2]{\begin{bsmallmatrix}
  #1\\
  #2
\end{bsmallmatrix}}
\newcommand{\IN}{{{X}}}
\newcommand{\OUT}{{{Z}}}
\newcommand{\truOUT}{{L_{{\nu}}^2}}
\newcommand{\truIN}{{L_{{\mu^{\prime}}}^2}}
\newcommand{\spIN}{{\mathcal{H}_{\IN}}}
\newcommand{\spOUT}{{\mathcal{H}_{\OUT}}}
\newcommand{\spOUTx}{{\mathcal{H}_{{X}}}}

\newcommand{\projP}{\widetilde{P}}

\newcommand{\Kreg}{\bm{K}_{\reg}}

\DeclareMathOperator*{\argmin}{arg\,min}

\usepackage{wrapfig}

\definecolor{TUMblue}{RGB}{0, 101, 189}
\definecolor{TUMwhite}{RGB}{255, 255, 255}
\definecolor{TUMblack}{RGB}{0, 0, 0}
\tikzset{%
  bullet/.style={ yshift = -.1mm
    , scale = .8},
  extfill/.style={fill = TUMblue!20}, %
  frozenfill/.style={fill = gray, fill opacity = .4}, %
  ext/.style={trapezium, trapezium angle=67.5, draw,
  inner ysep=5pt, outer sep=0pt, extfill,
  minimum height=1.2cm, minimum width=2pt},
  extl/.style={isosceles triangle,
    isosceles triangle apex angle=60, minimum height=1.2cm,
    draw, extfill, minimum width=0pt},
  bias/.style={minimum height=0pt,
    draw, extfill, minimum width=2.4em},
  block/.style    = {draw, thick, rectangle, minimum height = 3em, minimum width = 5em},
  sum/.style      = {draw, circle, node distance = 2cm}, %
  con/.style      = {draw, circle, node distance = 2cm}, %
  input/.style    = {coordinate}, %
  int/.style    = {coordinate}, %
  missing/.style={
    draw=none,
    fill=none,
    yshift = 0.1cm,
    scale = 1.5,
    execute at begin node={\color{black}{$\vdots$}}
  },
  output/.style   = {coordinate}, %
  exts/.style={ext, 
  minimum height=2cm},
  pini/.style = {pin edge={to-,thick,black}}, %
  pino/.style = {pin edge={-to,thick,black}}, %
  blk/.style = {draw, thick, fill=TUMblue!20, minimum size=2.8em, minimum width = 6em, text centered, text width = 9em},
  theblk/.style = {blk, text width = 5em, minimum height = 4em, fill = none, text centered},
}

\usepackage{microtype}
\usepackage{algorithm}
\usepackage{algpseudocode}
\algrenewcommand\textproc{}           % no "Procedure" label
\algrenewcommand\algorithmicrequire{\textbf{Input:}}
\algrenewcommand\algorithmicensure{\textbf{Output:}}


\algrenewcommand\alglinenumber[1]{\footnotesize\texttt{\textsc{A#1}}}

\makeatletter
\patchcmd{\thebibliography}%
  {\section*{\color{black}}}%
  {\section*{\color{black}References}}%
  {\typeout{cKOR: bibliography heading patched}}%
  {\typeout{cKOR: bibliography heading patch FAILED}}%
\makeatother
\usepackage{generic}

\newbool{tacFrontMatter}
\booltrue{tacFrontMatter}   % \boolfalse{tacFrontMatter} to switch off
\ifbool{tacFrontMatter}{%
  \markboth{\hskip25pc IEEE TRANSACTIONS ON AUTOMATIC CONTROL}
  {Bevanda \MakeLowercase{\textit{et al.}}: Nonparametric Control Koopman Operators}%
  \pubid{\scriptsize 0018-9286~\copyright~2026 IEEE. Personal use is permitted, but republication/redistribution requires IEEE permission. See https://www.ieee.org/publications/rights/index.html for more information.}%
}{}

\begin{document}
\title{\textcolor{black}{Nonparametric Control Koopman Operators}}
\author{Petar Bevanda, Bas Driessen, Lucian Cristian Iacob, Stefan Sosnowski, Roland T\'{o}th and Sandra Hirche
\thanks{This work was supported by the European Union’s Horizon program under Grant No. 101093822 (SeaClear2.0); by the Deutsche Forschungsgemeinschaft (DFG) under Grant No. 535860958 (ALeSCo); and by the European Research Council Consolidator Grant No. 864686 (CO-MAN).}
\thanks{P.~Bevanda, S.~Sosnowski, and S.~Hirche, are with Chair of Information-oriented Control, TU M\"{u}nchen, Germany. {\tt{\small \{petar.bevanda, sosnowski, hirche\}@tum.de}}. S.~Hirche is also affiliated with the Munich Data Science Institute (MDSI) and the Munich Institute of Robotics and Machine Intelligence (MIRMI), TU M\"{u}nchen, Germany.}
\thanks{L.C.~Iacob, and R.~T\'{o}th  are with the Control Systems Group, TU/e,
Eindhoven, The Netherlands {\tt{\small \{l.c.iacob@,r.toth@\}tue.nl}}. R.~T\'{o}th is also affiliated with the Systems and Control Laboratory, HUN-REN Institute for Computer Science and Control, Budapest, Hungary.}
\thanks{Digital Object Identifier 10.1109/TAC.2026.3693027}}
\maketitle
\pubidadjcol
% REQUIRED
\begin{abstract}
This paper presents a novel Koopman composition operator representation framework for control systems in reproducing kernel Hilbert spaces (RKHSs) that is free of explicit dictionary or input parametrizations. By establishing fundamental equivalences between different model representations, we are able to close the gap of control system operator learning and infinite-dimensional regression, enabling various empirical estimators and the connection to the well-understood learning theory in RKHSs under one unified framework. Consequently, our proposed framework allows for arbitrarily accurate finite-rank approximations in infinite-dimensional spaces and leads to finite-dimensional predictors without a priori restrictions to a finite span of functions or inputs. To enable applications to high-dimensional control systems, we improve the scalability of our proposed control Koopman operator estimates by utilizing sketching techniques. Numerical experiments demonstrate superior prediction accuracy compared to bilinear EDMD, especially in high dimensions. Finally, we show that our learned models are readily interfaced with linear-parameter-varying techniques for model predictive control. 
\end{abstract}
% REQUIRED
\begin{IEEEkeywords}
Data-driven modeling, Nonlinear control systems, Kernel methods, Machine learning, Koopman operators, Reproducing kernel Hilbert spaces
\end{IEEEkeywords}
\preSection
\section{Introduction}
\IEEEPARstart{R}{ecent} years have seen an ever-growing interest across different fields in constructing operator-theoretic models that can provide \textit{global} insight into physical or biological characteristics of observed phenomena \cite{Engel2000}, facilitating tractable analysis and control design for nonlinear dynamics \cite{BEVANDA2021197,Brunton2022,Otto2021,Busidic2012}. While, historically, modeling insight based on \textit{first principles} form physics was the driving force in modeling, %the 
increasing system complexity %of observed phenomena
\cite{MEZIC2004101,Mezic2005} limits their utility for %accurate 
modeling in engineering, necessitating  the use of data-driven methods. A promising framework that has recently reemerged and gained traction in the data-driven modeling community is based on the Koopman operator \cite{Koopman1931}, whose spectral decomposition can enable linear superposition of signals for possibly highly nonlinear systems \cite{Busidic2012}. This representational simplicity of dynamics inspired a bevy of applications in system identification \cite{Li2017,Otto2019, Azencot2020}, soft robotics \cite{Bruder2021b,Haggerty2023}, optimal control \cite{Villanueva2021, Guo2022,Houska2023}, to name a few.

\subsubsection*{Koopman-based representations for control systems}
As the Koopman framework was originally developed for autonomous systems, to accommodate control inputs, different methods have been proposed. These range from heuristically selecting a linear time-invariant (LTI) model class \cite{KORDA2018149}, having a finite set of input values and describing a switched model \cite{PEITZ2019184}, or analytically deriving the lifted representation \cite{Goswami2017}.
It has become established that control-affine systems can be written as bilinear lifted models under certain conditions, at least in continuous-time. The authors of \cite{IACOB2024} show that for both continuous- and discrete-time systems with inputs, an invariant Koopman form can be analytically derived, granted that the autonomous part is exactly embedded. The resulting model class contains a state and input-dependent input contribution, which is often not bilinear, especially in the discrete-time case. Thus, a globally exact finite-dimensional representation generally requires a non-affine control input dependence or a recasting to a \textit{linear parameter-varying} (LPV) model form. While this has been shown on an analytic level for finite-dimensional Koopman operator-based representations, it is an open question whether nonlinear input terms are unavoidable in the infinite-dimensional case and if approximation errors could be handled under certain but general assumptions.
\begin{table*}[t!]
% \scriptsize
  \caption{Comparison of existing {linear operator learning} approaches for control systems. Our cKOR approach is based on risk minimization and works with infinite-dimensional autonomous and control dependence.}\label{tab:comp}
\centering 
  \begin{tabular}{l|ccccccc} \toprule
   \textsc{Approach} & \txt{controls} &  $\dim(\txt{input space})$ &  $\dim(\txt{output space})$ &  $|\txt{datasets}|$& \txt{risk notion} & $\txt{ERM}$ & \txt{scalability}  \\ \midrule
      switched $\bm{u}$ \cite{PEITZ2019184} & quantized & $\text{finite}$ & $\text{finite}$ &  ${n_{{u}}}{+}1$ & \xmark & \xmark & \xmark  \\ 
       bEDMD: \cite{Bruder2021,otto2023learning}; & arbitrary & $\text{finite}$ & $\text{finite}$ &  $1$  & \xmark & \xmark & \checkmark \\ 
$\{\text{EDMD}(\bm{u}_i)\}^{n_u}_{i=0}$\cite{Peitz2020,Nuske2023,Philipp2023b} & constant & $\text{finite}$  & $\text{finite}$ &  ${n_{{u}}}{ + }1$ & \xmark & \xmark & \xmark \\ 
    \bf \txt{cKOR} (ours) & \bf arbitrary & $\bm{\infty}$ & $\bm{\infty}$ & \bf 1 & $\geq\|{\mathrm{\txt{error}}}\|^2_{\mathrm{op}}$ & \bf \ding{52} & \bf \ding{52}  \\
    \bottomrule
  \end{tabular}
  \postFloat
  \postFloat
\end{table*}
\subsubsection*{Data-driven operator-based approaches for control} 
A number of deep learning-enhanced, yet parametric, methods \cite{Guo2023,Haseli2023,Bevanda2022} have been proposed to capture nonlinear data relations, but commonly lack rigorous learning-theoretic foundations. In contrast, kernel-based operator learning provides a powerful alternative \cite{BATLLE2024112549} that is mathematically and implementation-wise simple, but offers rigorously established avenues for approximation error analysis.
Unsurprisingly, the aforementioned has led to a recent increase in learning-theoretic understanding of nonparametric Koopman operator-based models \cite{kostic2022learning,kostic2023sharp,KKR_neurips2023,bevanda2025koopmanequivariant} for autonomous systems. Nevertheless, their control system counterparts do not enjoy a comparable level of understanding.
For example, Hilbert space Koopman operators for control systems are not defined to full generality in existing literature, often requiring restrictive arguments involving specific generator discretizations \cite{Philipp2023b} or state inflation \cite{KORDA2018149}. This impedes a connection to strong approximation results and learning in a flexible nonparametric (dictionary-free) manner. 

A major technical reason for the aforementioned theoretical gap can also be traced back to a rather practical desire for finite-dimensional models. However, early discretization of the learning problem inspired by finite-section methods, e.g., by using an explicit and fixed feature or input dependence \cite{Proctor2018,Proctor2016,Kaiser2021,NuskeRFF,Nuske2023}, leads to a systematic loss of precision and inefficient exploitation of the data \cite{Khosravi2023}. The use of data-independent finite-dimensional subspace is especially ill-suited when dealing with unknown large-scale systems that require a suitably large/rich feature or input space. As summarized in Table \ref{tab:comp}, existing data-driven operator-theoretic control system models do not enable input and output spaces to be jointly infinite-dimensional and require multiple operator regressions. 
\subsubsection*{Our approach}
To alleviate the theoretical and practical limitations mentioned above, we formalize a novel dictionary-free learning approach using reproducing kernel Hilbert spaces. We connect infinite-dimensional regression with composition operators of control systems to provide a rigorous and self-contained nonparametric learning framework. This turns out to be crucial to get hold of both the approximation error as well as avoiding explicit tensor products of the dictionary and the control inputs \cite{otto2023learning,STRASSER20232257}, leading to arbitrarily accurate operator approximation.
As our numerical experiments confirm, this has strong implications: our nonparametric approach significantly outperforms classically used (bilinear) EDMD approaches, which commonly aim at approximating finite-element methods from data \cite{KORDA2018149,Zhang2022QunatiAnal}. Moreover, the approach we propose does not rely on carefully crafted datasets or multiple regressions that can be found in existing works, such as \cite{PEITZ2019184,Nuske2023,Philipp2023b}.
\subsubsection*{Contribution} To connect to the existing body of work with bilinear representations, we first derive tensor product RKHSs (Theorem \ref{thm:CAkern}, Remark \ref{rmk:infU}), establishing equivalent operator-theoretic models (Corollary \ref{cor:LPV}).
By connecting regression risk (Lemma \ref{lem:riskBound}) to operator norm error and proving arbitrarily accurate control system operator approximation under minimal assumptions (Theorem \ref{thm:opNorm}), guarantees that any RKHS observable admits arbitrarily accurate prediction using finite-rank operators (Corollary \ref{coro:Opred}). Given a dataset, we turn finite-rank infinite-dimensional operators to closed-form finite-dimensional predictors (Proposition \ref{prop:cKORkrrSHORT} and \ref{prop:NYcKORkrrSHORT}).
Finally, we statistically confirm the advantage of our nonparametric framework on various prediction tasks and show that our models can be directly used in computationally efficient iterated LPV model predictive control (MPC) methods.
\subsubsection*{Notation}
For non-negative integers $n$ and $m$, $[m,n]=\{m,\dots,n\}$ with $n\geq m$ gives an interval set of integers. We use the shorthand $[n]{\defeq}[1,n]$. Given two separable Hilbert spaces $\mathcal{F}$ and $\mathcal{Y}$, we let $\HS{\mathcal{F},\mathcal{Y}}$ be a Hilbert space of Hilbert-Schmidt (HS) operators from $\mathcal{F}$ to $\mathcal{Y}$ endowed with the norm $\|{A}\|_{\HS{\mathcal{F},\mathcal{Y}}}^2\equiv \lilsum_{i\in\Set{N}}\norm{Ae_i}^2_{\mathcal{Y}}=\mathrm{Tr}(A^*A)$, where $\{e_i\}_{i\in\Set{N}}$ is an orthonormal basis of ${\mathcal{F}}$. For HS-operators from $\mathcal{F}$ to itself, we use the shorthand $\HS{\mathcal{F}}$. The operator norm of a linear operator ${\mathcal{G}}:\mathcal{F} \to \mathcal{Y}$ is denoted as $\|{\mathcal{G}}\|_{\mathrm{op}}\defeq \sup _{\|f\|_{\mathcal{F}}=1}\|{\mathcal{G}} f\|_{\mathcal{Y}}$, with image $\operatorname{Im}(\mathcal{G})\defeq\mathcal{G}(\mathcal{F})$ and $\operatorname{cl}(\cdot)$ its norm closure. ${\mathcal Y} \otimes {\mathcal F}$ denotes the Hilbert‑space completion of the algebraic tensor product ${\mathcal Y} \otimes_{\mathrm{alg}} {\mathcal F}$ with inner product
$\langle y_1 \otimes f_1,\; y_2 \otimes f_2\rangle
   = \langle f_1,f_2\rangle_{\mathcal F}\,
     \langle y_1,y_2\rangle_{\mathcal Y}$, and we regard  $y \otimes f$ as elementary tensors of this Hilbert‑space tensor product. A direct sum of Hilbert spaces $\mathcal{Y}$ and $\mathcal{F}$ is denoted as $\mathcal{Y}\oplus \mathcal{F}$.
For simplicity, adjoints of operators as well as (conjugate) transposes of matrices are denoted as $\adjoint{(\cdot)}$. 
Lower/upper case symbols denote functions/operators, bold symbols are reserved for matrices and vectors, and $\odot$ denotes the Hadamard product.
The space of square-integrable functions is denoted as $L_{{{\mu}}}^2(\cdot)$ with an appropriate Lebesgue measure ${{\mu}}$ while the vector space of bounded continuous functions is denoted by ${C
}_0(\cdot)$. On a domain $\Set{X}$ let $k_X: \Set{X} \times \Set{X} \rightarrow \Set{R}$ be a symmetric and positive definite kernel function and $\RKHS_X$ the corresponding RKHS \cite{IngoSteinwart2008SupportMachines}, with norm denoted as $\|\cdot\|_{\RKHS_X}=\sqrt{\langle \cdot,\cdot\rangle_{\RKHS_X}}$.  For $\bm{x}\in\Set{X}$, ${\phi}_{X}(\bm{x}) \equiv k_X(\cdot, \bm{x}) \in \RKHS_X$ denotes the canonical feature map.
\subsubsection*{Organization} After the preliminaries, Section \ref{sec:InfProb} states the infinite‑dimensional learning problem we intend to address. Section \ref{sec:InfSolution} solves it at the \textit{population} (infinite-data) level and links it to various kernel choices. Then, in Section \ref{sec:cKOR}, finite‑sample predictors are derived, while Section \ref{sec:Sketch} adds Nystr\"{o}m variants to reduce the computational load. Section \ref{sec:POD} provides efficient methods for state-dimension reduction, while in Section \ref{sec:NumEx} the numerical validation of our proposed approach is provided. Finally,  while Section \ref{sec:Concl} and appendices provide conclusions and additional details, respectively.
\section{Preliminaries and Problem Statement}\label{sec:probStat}
To put our developments into perspective, we introduce two relevant and distinct modeling paradigms.
\subsection{Modeling Paradigms}
First, we outline the classical state-space modeling paradigm, ubiquitous in control systems theory. 
\subsubsection*{State-Space} Consider a nonlinear system 
\begin{equation}\label{eq:ncs}
    \bm{x}_{k{+}1}=\bm{\mathsf{f}}(\bm{x}_{k},\bm{u}_{k})  \tag{\txt{\textsc{ncs}}}
\end{equation}
where $\bm{x}_{k}\in \Set{X} \subset \Set{R}^{n_{{x}}}$ is the state and $\bm{u}_{k} \defeq [u_1({k})\cdots u_{n_{{u}}}({k})]^\top \in  \Set{U} \subset \Set{R}^{n_{{u}}}$ denotes the control variable and $k\in\Set{N}_0^+$ is the discrete time. Throughout, we consider $\bm{\mathsf{f}}$ to be continuously differentiable and the sets $\Set{X}$ and $\Set{U}$ to be compact, ensuring global Lipschitz continuity on $\Set{X}\times\Set{U}$.
\subsubsection*{Operator-Theoretic}
The existence of \eqref{eq:ncs} allows one to capture more than the state-space dynamics, i.e., to {describe the dynamics of Hilbert-space-valued functions} (\textit{observables}) over the system \eqref{eq:ncs} via a \textit{linear operator} between \textit{separable} Hilbert spaces.
The structure of this operator and of the spaces it acts between is, however, not canonical: the operator-theoretic Koopman-based control literature \cite{Goswami2017,otto2023learning,Otto2022LearningTrajectories,Goswami2022BilinearizationApproach,Bruder2021} has predominantly settled on a specific, bilinear, tensor-product assumption. In particular, inspired by the widespread control-affine special case $\bm{\mathsf{f}}(\bm{x},\bm{u}) = \bm{f}(\bm{x}) + \lilsum^{n_{{u}}}_{j=1} \bm{g}_j(\bm{x})\, u_j$ \cite{Nijmeijer96}, a linear operator model ${\mathcal{G}}: \RKHS^\prime \to \RKHS$ is considered, obeying
\begin{align}\label{eq:dualBilin}
y_{k+1}(\bm x_k,\bm u_k)
\;\defeq\;
\langle{ \mathcal G\,y_k},
  { \concat{1}{\bm u_k}\otimes \phi(\bm x_k)}\rangle_{\mathcal H} \tag{\txt{\textsc{bts}}}
\end{align}
for any observable $y \in \RKHS^\prime$, where the basis $\concat{1}{\bm{u}}\otimes \phi (\bm{x}) \in \RKHS$ and ${\phi}(\bm{x}) \in \RKHS^\prime$ span $\RKHS$ and $\RKHS^\prime$, respectively\footnote{Any bounded bilinear control system on a Hilbert space can be rewritten as a linear one on the tensor product space \cite{Isidori_Ruberti_1973,Elliott_2009}.}.
We generalize beyond this bilinear assumption to arbitrary square-integrable functions of state and control, where the operator \textit{maps between fundamentally different spaces}, contrasting the classical Koopman operator setting.
 \begin{definition}\label{def:cko}
 The \textit{Control Koopman Operator} is a linear operator ${\mathcal{G}}: \truIN(\Set{X})\to \truOUT(\Set{X} {\times} {\Set{U}})$ defined by the composition
     \begin{align}\label{eq:cKOp}
   [{\mathcal{G}}y](\bm{x},\bm{u}) = y \circ \bm{\mathsf{f}}(\bm{x},\bm{u}) \tag{\txt{\textsc{cko}}},
\end{align}
of an observable $y \in \truIN$ with \eqref{eq:ncs} for all $(\bm{x},\bm{u}) \in  \Set{X} \times \Set{U}$.
 \end{definition}
\subsection{Linear-in-Parameter Regression}
\subsubsection*{Function-to-Vector}
An established problem in data-driven system identification is the regression of \eqref{eq:ncs} from observing tuples $\{(\bm{x}_k,\bm{u}_k), \bm{x}_{k+1}\}$ on $\Set{Z}{\times}\Set{X}$ with $\Set{Z}\defeq\Set{X}{\times}\Set{U}$, coming from some unknown data distribution $\bm{z}_k\defeq(\bm{x}_k,\bm{u}_k) \sim \nu$ in the form of snapshots
\begin{equation}\label{eq:SRdataset_snapshot_pairs}
    \bigl\{\phi_{Z}(\bm{z}^{(i)}) \defeq \phi_{Z}(\bm{z}^{(i)}_{k})~,~~\bm{x}_{+}^{(i)} \defeq \bm{x}^{(i)}_{k{+}1}\bigr\}_{i{\in}[n]},
\end{equation}
with the real-valued hypothesis $\spOUT=\{\phi_{Z}(\bm{z})\}$ a separable (possibly infinite-dimensional) Hilbert space.
Classically, assuming the hypothesis $\RKHS$ provides a suitably rich parameterization to represent the unknown mapping $\bm{\mathsf{f}}$ \eqref{eq:ncs}, the identification problem is well-specified and usually (assuming existence) follows from scalar-valued \textit{risk minimizations}
\begin{align}\label{eq:scalar-regression}
  \Bigl\{\min_{\theta \in \spOUT}\mathbb{E}\left[\|\bm{\mathsf{f}}_d(\bm{z})-\innerprod{\theta_d}{\phi_{Z}(\bm{z})}_{\spIN}\|_{\Set{R}}^2\right]\Bigr\}_{d \in [n_{x}]},
\end{align}
for individual components of the state $d\in[n_{x}]$ \cite{Umlauft2017,castañeda2021gaussian,Lederer2021TheControl}.
\subsubsection*{Function-to-Function}
Contrasting the classical system identification problem for the dynamics in state-space $\Set{X}$, we will ``lift" the target label $\bm{x}_+ \mapsto {\phi}_{X}(\bm{x}_+)$ to an auxiliary hypothesis space $\spIN=\{{\phi}_{X}(\bm{x}_+)\}$, turning \eqref{eq:SRdataset_snapshot_pairs} into
\begin{equation}\label{eq:dataset_snapshot_pairs}
    \bigl\{\phi_{Z}(\bm{z}^{(i)}),{\phi}_{X}(\bm{x}_{+}^{(i)}) \bigr\}_{i{\in}[n]}. 
\end{equation}
Our objective in this work is to learn a function-to-function -- \textit{operator} -- mapping \cite{Steffen2012}
\begin{align}\label{eq:OpRegFeature}
  \min_{G \in \HS{\RKHS_X,\RKHS_Z}}~\mathbb{E}\left[\|{\phi}_{X}(\bm{\mathsf{f}}(\bm{z}))-G^*\phi_{Z}(\bm{z})\|_{\spIN}^2\right].
  \tag{\txt{\textsc{f-f}}}
\end{align}
The appeal of operator-based modeling comes from the fact that once a finite-rank operator is obtained, any function that belongs to $\spIN$ can be predicted. To treat both \eqref{eq:scalar-regression} and \eqref{eq:OpRegFeature} on equal footing, we recognize that the set of scalar-valued regressions of \eqref{eq:scalar-regression} equals
\begin{align}\label{eq:finite-regression}
\tag{\text{\textsf{\textsc{f-v}}}}
    \min_{\theta \in \HS{\Set{R}^{n_{x}},\spOUT}}~\mathbb{E}\left[\|\bm{\mathsf{f}}(\bm{z})-{\theta^*}{\phi_{Z}(\bm{z})}\|_{\Set{R}^{n_{x}}}^2\right],
\end{align}
by the natural isometric isomorphism $ (\theta_1,\dots, \theta_{n_{x}}) =  \theta \in (\spOUT)^{n_{x}} \cong \HS{\Set{R}^{n_{x}},\spOUT} $ \cite{aubin2011applied}.
Therefore \eqref{eq:OpRegFeature} can be seen as a direct generalization of \eqref{eq:finite-regression} to possibly infinite-dimensional targets spanning $\spIN$, in turn, allowing the prediction of any $y \in \spIN$ at inference.
\subsection{Problem Statement}
We aim to address the problem of learning the control Koopman operator ${\mathcal{G}}$ from Definition \ref{def:cko} when observed through the samples of its action on a hypothesis space. Connecting to feature-to-feature regression, we formulate the operator-theoretic equivalent of \eqref{eq:OpRegFeature} for \eqref{eq:cKOp} using its \textit{restriction ${\mathcal{G}}{{\mid_\spIN}}$ to a reproducing kernel Hilbert space (RKHS)} where the approximation
\begin{enumerate}
    \item Follows from infinite-dimensional regression;
    \item Avoids the curse of input dimensionality;
    \item Allows for control-uniform prediction of observables;
    \item Enables with well-posed estimation schemes;
    \item Rigorously derive finite-dimensional models;
    \item Scales to large datasets using sketching \cite{Williams2000,ahmad2023sketch}.
\end{enumerate}
\preSection
\section{Operator Learning in Infinite-Dimensions}\label{sec:InfProb}
Our approach is rooted in the well-studied theory of infinite-dimensional regression in (reproducing kernel) Hilbert spaces, leading to an RKHS-valued regression problem for learning control Koopman operators in a flexible yet principled manner. 
\subsubsection*{RKHSs as subspaces} We consider RKHSs $\spIN$/$\spOUT$ that are a subset of $\truIN$/$\truOUT$-integrable functions \cite[Chapter 4.3]{IngoSteinwart2008SupportMachines} with associated canonical feature maps ${\phi}_{X}: \Set{X} \rightarrow \spIN$ and ${\phi_{Z}}: \Set{Z} \rightarrow \spOUT$, so that we can approximate the \eqref{eq:cKOp} \textit{restriction} 
\begin{subequations}
    \begin{equation}
        \mathcal{T}\defeq {\mathcal{G}}{_{\mid_\RKHS}}: \spIN\to \truOUT(\Set{Z}) \qquad \approx \qquad G: \spIN \rightarrow \spOUT
    \end{equation}
\begin{equation}
    % \adjustbox{scale=1.,center}{%
    \begin{tikzcd}[row sep=3em,column sep=5em]
     \truIN \arrow[r, "{\mathcal{G}}", ] 
    \arrow[d, leftarrow, "{S}_{\mu^{\prime}}"']
    & |[]| \truOUT  \arrow[thick,d, leftarrow, "{S}_{\nu}"] 
    \\
    |[]| \spIN \arrow[thick,r, "G"] \arrow[thick, ru, "\mathcal{T}", above, shift left]
    % \arrow[ru, "{S}_{\nu}{G}"', below, shift right ]
    & \spOUT
    \end{tikzcd}
    % }
\end{equation} 
\end{subequations}
where we account for the norm discrepancy between square-integrable and RKHS functions using \textit{inclusions} 
\begin{subequations}\label{eq:incl}
       \begin{align}
S_{\nu}:\spOUT\hookrightarrow\truOUT\quad\mathrm{s.t.}\quad & \spOUT \ni g \mapsto [g]_{\sim} \in \truOUT,\\
{S}_{\mu^{\prime}}:\spIN\hookrightarrow\truIN\quad\mathrm{s.t.}\quad & \spIN \ni f \mapsto [f]_{\sim} \in \truIN,
\end{align}  
\end{subequations}
which transform RKHS elements to their point-wise equal equivalence class $[\cdot]_{\sim}$ endowed with an appropriate $L^2_{({\cdot})}$-norm. Moreover, the inclusions admit adjoints
\begin{subequations}
           \begin{align}
           S^*_\nu: \truOUT \ni    g & \quad \mapsto \quad \Set{E}_{\bm{z}\sim \nu}\bigl[g(\bm{z}) \phi_{Z}({\bm{z}})\bigr]  \in \spOUT, \\
           S^*_{{\mu^{\prime}}}: \truIN \ni  f & \quad \mapsto \quad  \Set{E}_{\bm{x} \sim \mu^{\prime}}\bigl[f(\bm{x}) {\phi}_{X}(\bm{x})\bigr]\in \spIN.
       \end{align}
\end{subequations}
 \begin{assumption}\label{asm:Bcko}
     To ensure well-posedness, we require that \eqref{eq:cKOp} is bounded, i.e., $\exists \txt{L},~\mathrm{s.t.}~\| \mathcal{G}f \|_{\truOUT} \le \txt{L} \, \| f \|_{\truIN}$.
 \end{assumption}
\begin{remark}\label{rmk:Bound}
For Lipschitz unforced dynamics \(\bm{f}(\bm{x})\equiv\bm{\mathsf{f}}(\cdot,\bm{0})\), the above assumption is readily satisfied as they are locally invertible $\abs{\mathrm{det}(\nabla\bm{{f}}(\bm{x}))} \neq 0$ (singularity-free) almost-everywhere, describing the transient behavior of large classes of cyber-physical systems \cite{KKR_neurips2023,bevanda2025koopmanequivariant}. This extends to \eqref{eq:ncs} though Lipschitz continuity and the compactness of $\mathbb{X}$, $\mathbb{U}$, allowing a bounded \textit{measure distortion} $\nu(\{{\bm{z}}\in {\mathbb{X}\times\mathbb{U}}:\bm{\mathsf{f}}(\bm{z}) \in \mathbb{A}\})\le \txt{L}^2{\mu^{\prime}}(\mathbb{A})$ for any Borel set ${\mathbb{A}} \subseteq \mathbb{X}$, in turn, satisfying Assumption \ref{asm:Bcko}.
\end{remark}
To avoid being caught up in measurability and integrability issues, we impose the following technical assumptions.
\begin{assumption}
\label{asm::RKHS}
$\spIN,\spOUT$ are separable, and their kernel functions continuous and ${\mu^{\prime}}$- and $\nu$-a.e. bounded, respectively.
\end{assumption}
The above assumption is not restrictive, as, e.g., Gaussian, Laplacian, or Matérn kernels \cite{IngoSteinwart2008SupportMachines} satisfy all of the above assumptions over Euclidean domains \cite{li2022optimal}. 
Moreover, is well-established that, under Assumption \ref{asm::RKHS}, the inclusion operators \eqref{eq:incl} are Hilbert-Schmidt (HS) \cite[Chapter 4.3]{IngoSteinwart2008SupportMachines}, so 
$\mathcal{T} \in \HS{\spIN,\truOUT}$ follows by Assumption  \ref{asm:Bcko} \cite{Mollenhauer2020}.
\begin{figure*}[t!]
  \centering
   \resizebox{.7\textwidth}{!}{%
\begin{tikzpicture}[thick, node distance=5cm,auto]
    \node [blk] (x) {{\textsf {Learning}} \vspace{0.5em}\\$\adjoint{G}(\concat{1}{\bm{u}} \otimes {\phi}_{X}(\bm{x}))$};     %
    \node [blk] (c) [right of=x] {{\textsf {Simulation}} \vspace{0.5em}\\$(A^*+M^*({\bm{u}})){\phi}_{X}(\bm{x})$};
    \node [blk] (y) [right of=c] {{\textsf {Predictive Control}} \vspace{0.5em}\\$A^*{\phi}_{X}(\bm{x}){+}B^*_{{\phi}_{X}(\bm{x})}{\bm{u}}$};
    \path[<->] (x) edge node [pos=.5, above, scale = .8] {}  (c); %
    \path[<->] (c) edge node [pos=.5, above, scale = .8] {}  (y); %
\end{tikzpicture}
}
   \caption{
    Equivalent control operator-induced RKHS embeddings of the system dynamics for various tasks for the \textit{control-affine} range $\RKHS_Z$.}
  \label{fig:Equiv}
  \postFloat
\end{figure*}

\subsubsection*{Operator $\cong$ feature-to-feature regression} 
To learn $G$, we can define the feature-to-feature risk \eqref{eq:OpRegFeature}
\begin{align}\label{eq:risk}
    \Set{E}_{\bm{z} \sim \nu} \!\!\left[{\norm{{\phi}_{X}({\bm{\mathsf{f}}(\bm{z})})- G^{*} \phi_{Z}({\bm{z}})}^{2}_{{\spIN}}}\right],
\end{align}
describing a mean square error of a linear map ${G}^*$ from $\phi_{Z}({\bm{z}}) \in \spOUT$ to ${\phi}_{X}({\bm{\mathsf{f}}(\bm{z})}) \in \spIN$ that can be interpreted as the \textit{embedding} of \eqref{eq:cKOp} in RKHS
\begin{align}
    \left[\mathcal{T} h\right](\bm{z}) =\innerprod{h}{\mathcal{T}^*S_{\nu}\phi_{Z}({\bm{z}})}_{\spIN} =\langle h,{\phi}_{X}({\bm{\mathsf{f}}(\bm{z})}\rangle_{\spIN}.
\end{align}
By Assumptions 1 and 2, the above expression is well-defined. 
An application of Tonelli's theorem to swap orders then gives that $\Set{E}[\sum_i\langle h_i,{\phi}_{X}({\bm{\mathsf{f}}(\bm{z})}){-}G^*\phi_{Z}({\bm{z}})\rangle^2]$ equals $\sum_i \Set{E}[\langle h_i,{\phi}_{X}({\bm{\mathsf{f}}(\bm{z})}){-}G^*\phi_{Z}({\bm{z}})\rangle^2]$ for an orthonormal basis $\{h_i\}_{i \in \Set{N}}$ of $\spIN$. After using the Hilbert-Schmidt norm definition, it becomes apparent that \eqref{eq:risk} amounts to
    \begin{align}\label{eq:HSrisk}
  \mathsfit{R}(G) \defeq \left\|\mathcal{G}{-} G\right\|_{\mathrm{HS}(\spIN{,} {L}_{\nu}^2)}^2 \equiv \left\|\mathcal{T}{-}S_{\nu} G\right\|_{\mathrm{HS}}^2, \tag{\txt{\textsc{risk}}}
\end{align}
leading to the operator-theoretic risk minimization
\begin{align}\label{eq:MSEriskmin}
   \min_{G \in \HS{\spIN,\spOUT}} \mathsfit{R}(G) \equiv \left\|\mathcal{T}{-}S_{\nu} G\right\|_{\mathrm{HS}}^2,\tag{\txt{\textsc{ckor}}}
\end{align}
we refer to as \textit{control Koopman operator regression} (cKOR).
By the Pythagorean theorem, \eqref{eq:HSrisk} equals
\begin{align}\label{eq:riskEco}
 \underbrace{\norm{P_{\spOUT}\mathcal{G}{-}{G}}^{2}_{\mathrm{HS}(\spIN{,}{L}_{\nu}^2)}}_{\text{projected risk}}{+}\underbrace{\norm{[I{-}P_{\spOUT}]\mathcal{G}}^{2}_{\mathrm{HS}(\spIN{,}{L}_{\nu}^2)}}_{\text{representation risk}}, 
\end{align}
where $P_{\spOUT}$ is the orthogonal projector in $\truOUT$ onto $\spOUT$. The above decomposition is classical in learning theory \cite{IngoSteinwart2008SupportMachines, Bach2023LearningPrinciples} and operator regression \cite{Mollenhauer2020,Mollenhauer2022,kostic2022learning,kostic2023sharp}. While the \textit{projected risk} depends on the learned $G$, mitigating the $\spOUT$-dependent, \textit{representation risk} is crucial for achieving statistical consistency. Luckily, solving \eqref{eq:MSEriskmin} using suitably defined infinite-dimensional RKHSs \cite{IngoSteinwart2008SupportMachines}, the \textit{representation risk} can vanish, enabling arbitrary accurate learning and, in turn, prediction of $\spIN$-observables under control inputs. Moreover, any continuous observables can be represented to arbitrary accuracy by their embedding in universal RKHSs, covering a large class of observables whose evolution can be approximated while providing a surrogate model for the state-space dynamics \eqref{eq:ncs} as a special case.
\section{Nonparametric Control Koopman Operator Approximations}\label{sec:InfSolution}
The \textit{flexibility} of a nonparametric approach comes from the fact that, given a universal kernel, neither the feature map nor the feature space is uniquely determined -- defining a dictionary-free approach \cite{IngoSteinwart2008SupportMachines}. Nonetheless, a \textit{reproducing kernel Hilbert space} (RKHS) uniquely defines a kernel (and vice versa) \cite{IngoSteinwart2008SupportMachines}, so we impose the structure of \eqref{eq:dualBilin} on the feature space of an RKHS to derive a kernel that corresponds to it. As it turns out, such a kernel is crucial for practically working with infinite-dimensional spaces and getting a hold on the approximation error of the evolution of observables in $\spIN$. In the conditional expectation operator setting, it is established that working in infinite-dimensional RKHSs has various benefits, e.g., overcoming the weak convergence results \cite{KORDA2018149,Peitz2020} of finite-section methods, cf. Mollenhauer and Koltai \cite{Mollenhauer2020} for a discussion. In the context of linear operator learning for control systems, a similarly flexible framework is missing, which we propose here.
\subsection{Reproducing kernel Hilbert space representations}
By the structure of \eqref{eq:dualBilin}, tensor product spaces \cite{aubin2011applied,Mollenhauer2022} are of particular importance to endow the Hilbert space-valued Koopman operators with control effects in a principled manner. 
\subsubsection*{Hilbert tensor products}
For $y \in \mathcal{Y}$ and $f \in \mathcal{F}$, the bounded operator $y \otimes f \in \mathcal{L}(\mathcal{F}, \mathcal{Y})$ is the \emph{rank-one operator}
\begin{equation}\label{eq:rank-one_operator}
 \mathcal{F} \ni h \mapsto [y \otimes f](h) \defeq\innerprod{f}{ h }_{\mathcal{F}}y \in \mathcal{Y} .
\end{equation}
The Hilbert tensor (outer) product $\mathcal{Y} \otimes \mathcal{F}$ is defined to be the completion of the linear span of all such
rank-one operators with respect to the inner product 
\begin{align}\label{eq:TPrule}
    \innerprod{ y \otimes f }{ f^\prime \otimes y^\prime }_{\mathcal{Y} \otimes \mathcal{F}} \defeq\innerprod{ f}{ f^\prime }_{\mathcal{F}} \innerprod{ y }{ y^\prime }_{\mathcal{Y}}.
\end{align}
We will interchangeably use the isometric isomorphisms $\HS{\mathcal{F}, \mathcal{Y}} {\cong} \mathcal{Y} \otimes \mathcal{F}$ and $L^{2}_{{{\mu^{\prime}}}}(\mathcal{Y}) {\cong} L^{2}_{{{\mu^{\prime}}}}( \Set{R}) \otimes \mathcal{Y}$ \cite[Chapter~12]{aubin2011applied}, and treat such spaces as essentially identical.
\subsubsection*{Kernel functions} To appreciate the above tensor product construction, we start by recognizing that the dynamics \eqref{eq:dualBilin} satisfy the following pairing for any ${y} \in \truIN$
\begin{align}\label{eq:TensDyn}
    \innerprod{y}{{{\mathcal{G}}}^*(\concat{1}{\bm{u}_{k}} \otimes \phi (\bm{x}_{k}))}_\truIN   = \innerprod{{{\mathcal{G}}}y}{\concat{1}{\bm{u}_{k}} \otimes \phi (\bm{x}_{k})}_\truOUT
\end{align}
for a bounded \eqref{eq:cKOp}, revealing that the image space of ${\mathcal{G}}$ is essentially that of vector-valued $\concat{1}{\bm{u}_{k}}$-affine $\truOUT$-functions, helping us uniquely define the image RKHS $\spOUT$ -- equivalently kernel $k_Z: \Set{Z} \times \Set{Z} \mapsto \Set{R}$ --  for our hypothesis $G: \spIN \rightarrow \spOUT$.
\begin{theorem}[Control-affine kernel]\label{thm:CAkern}
Let $\spIN$ be a separable RKHS with corresponding kernel $k_X: \Set{X}\times\Set{X} \mapsto \Set{R}$ and $\bm{v}\defeq\concat{{1}}{\bm{u}} \in \Set{V} \subseteq \Set{R}^{n_u\!{+}\!1}$. Then, the completion of $\bm{v}$-affine functions is the $\Set{V} \otimes \spIN$, which is defined by the kernel
    \begin{align}\label{eq:CAkern}
    {k}^{\textsf{ca}}_Z(\bm{z},\bm{z}^\prime) ={k}_X(\bm{x}^\prime,\bm{x})(1{+}\innerprod{\bm{u}}{\bm{u}^\prime})
\end{align}
that corresponds to RKHS $\spOUT$ and equivalently defines $\concat{1}{\bm{u}}$-affine $\spIN$-valued observables via the operator-valued kernel $K_{X}(\bm{x},\bm{x}^\prime)\defeq k_{X}(\bm{x},\bm{x}^\prime)\mathrm{Id}_{\Set{V}} \in \mathcal{L}(\Set{V})$, where $\mathcal{L}(\Set{V})$ is the set of bounded operators from $\Set{V}$ to itself.

\begin{proof}
Let $V\defeq K_{X}(\cdot,\bm{x}) \bm{v} ,V^{\prime}\defeq K_{X}(\cdot,\bm{x}^\prime) \bm{v}^\prime$ belong to vector-valued RKHS $\mathscr{G}$ of $ \bm{v}$-affine functions defined as
\begin{align}\label{eq:spanCA}
    \mathscr{G}=\overline{\mathrm{span}\{K_{X}(\cdot,\bm{x}) \bm{v}\mid \bm{v} \in \Set{V}, \bm{x}\in\Set{X}\}}^{\|\cdot\|_{\mathscr{G}}}
\end{align}
Thus, we have $\left\langle V, V^{\prime}\right\rangle_{\mathscr{G}}=\innerprod{ K_{X}(\cdot,\bm{x}) \bm{v}}{K_{X}(\cdot,\bm{x}^\prime) \bm{v}^\prime}_{\mathscr{G}}= \innerprod{  K_{X}^*(\cdot,\bm{x}^\prime)K_{X}(\cdot,\bm{x}) \bm{v}}{\bm{v}^\prime}_{\Set{V}}
=\innerprod{ k_{X}(\bm{x}^\prime,\bm{x}) \bm{v}}{\bm{v}^\prime}_{\Set{V}}= \innerprod{ \bm{v}^\prime \otimes  {\phi}_{X}(\bm{x}^\prime) }{\bm{v} \otimes  {\phi}_{X}(\bm{x})}_{\Set{V} \otimes \spIN}
=\innerprod{ {\phi}_{X}(\bm{x}^\prime) }{ {\phi}_{X}(\bm{x}) }_{\spIN}\innerprod{\bm{v}}{ \bm{v}^\prime}_{\Set{V}} =k^{\textsf{ca}}_Z(\bm{z},\bm{z}^\prime)$.
\end{proof}
\end{theorem}
          The structured kernel \eqref{eq:CAkern} should not be confused with parametric models that work with a finite set of observables. Namely, using a $C_0$-universal $\spIN$ guarantees that \eqref{eq:CAkern} induces an RKHS dense in the space of bounded continuous as well as square-integrable control-affine functions \cite{IngoSteinwart2008SupportMachines}. 
  \begin{remark}[Beyond control-affine $\RKHS_Z$]\label{rmk:infU}
  Theorem \ref{thm:CAkern} and subsequent results immediately hold for infinite-dimensional control kernels in place of $\innerprod{\bm{u}}{\bm{u}^\prime}$, e.g.,  by changing the kernel to ${k}_Z(\bm{z},\bm{z}^\prime)  ={k}_X(\bm{x}^\prime,\bm{x})(1{+}k_{U}({{\bm{u}},{\bm{u}^\prime}}))$  where $k_U({\bm{u}},{\bm{u}^\prime})$ may be an infinite-dimensional RKHS $\RKHS_U$. This is particularly useful for the case where an unstructured nonlinear system admits a control-affine reformulation \cite{Nijmeijer96}. This structure encodes that \eqref{eq:cKOp} has a natural evolution in the absence of input and allows one to interpret the two components separately w.l.o.g. Moreover, if $\RKHS_U$ is also $C_0$-universal, ${k}_X(\bm{x}^\prime,\bm{x})(1{+}k_U({{\bm{u}},{\bm{u}^\prime}}))$ induces an RKHS dense in the range of \eqref{eq:cKOp}.
  \end{remark}
 \subsubsection*{Representational equivalence}
 While our hypothesized dynamics are modeled by the linear operator $\Estim^*$, the explicit form of the simulation model is not immediately obvious due to an ``input-evolving" RKHS.
Intuitively, given a fixed control value would collapse the tensor product in \eqref{eq:TensDyn} and make the system autonomous in $\spIN$. We make such intuition rigorous, revealing equivalent linear parameter-varying (LPV) operator formulations of our hypothesis $G \in \HS{\spIN,\spOUT}$.
\begin{corollary}\label{cor:LPV}
Let $\spIN,\RKHS_U$ be separable and let $\RKHS^{1}_{U}\defeq{1}_U{\oplus}\RKHS_U$ with $1_U \defeq \mathrm{span}\{ f \}
\quad\text{where}\quad 
f:\Set{U} \to \mathbb{R},\;\; f(\bm{u}) \equiv 1 \;\; \forall \bm{u} \in \Set{U}$. Consider $\spOUT=\RKHS^{1}_{U}  \otimes  \spIN$ in $\adjoint{G} {\in} \HS{\spOUT{,}\spIN}$ with $\left(e_{i}\right)_{i \in \mathbb{N}}$ an orthonormal basis of $\RKHS_{U}$ and $\left(e^*_{i}\right)_{i \in \mathbb{N}}$ its dual basis. Then, the following isometry
     $\adjoint{G} \longleftrightarrow  \concat{1}{0}^*{ \otimes} A^*+\lilsum_{i \in \mathbb{N}} \concat{0}{e_i}^*{ \otimes} B^*\left({e_i}\right)$,
explicitly establishes the isometric isomorphisms between $\HS{\RKHS^{1}_{U}  \otimes  \spIN,\spIN} \cong {\RKHS^1_{U}}^*  \otimes  \HS{\spIN} = \HS{\spIN} \oplus \RKHS_U^*  \otimes  \HS{\spIN}$ where $A^*\defeq A^*\left(\concat{1}{0}\right)$ and $B^*(e_i)\defeq A^*\left(\concat{0}{e_i}\right)$. Moreover $\HS{\spIN} \oplus \RKHS_U^*  \otimes  \HS{\spIN} \cong  \HS{\spIN} \oplus \HS{\RKHS_U,\HS{\spIN}}$, inducing equivalent RKHS embeddings in Figure \ref{fig:Equiv}.

\begin{proof}
The isometric isomorphism is a direct consequence of \cite[Theorem 12.3.2. \& Proposition 12.3.1.]{aubin2011applied} where we factored out the control input-independent span. Applying the former results once more, we have that $\HS{\spIN} \oplus \RKHS_U^*  \otimes  \HS{\spIN} \cong  \HS{\spIN} \oplus \HS{\RKHS_U,\HS{\spIN}}$. To see the equivalence of representations in Figure \ref{fig:Equiv}, consider a finite-dimensional kernel $\innerprod{\bm{u}}{\bm{u}^\prime}$ for $\RKHS_U$ whose orthonormal basis is the standard basis $\left(\bm{e}_{i}\right)_{i \in [n_u]}$ of $\Set{R}^{n_u}$. Then, since
${\bm{u}}{=}\lilsum^{n_u}_{i=1}[\bm{e}_i \otimes \bm{e}_i]({\bm{u}}){=}\lilsum^{n_u}_{i=1}\left\langle \bm{e}_i^*,{\bm{u}}\right\rangle \bm{e}_i $, we have
\begin{subequations}
    \begin{align}
\!\!\!\!\!\!\adjoint{G}({\bm{v}}{ \otimes} {\phi}_{X}(\bm{x})) = (A^*&{+}\underbrace{[\lilsum^{n_u}_{i=1}\!B^*(\bm{e}_i)\left\langle\bm{e}_i^*,  {\bm{u}} \right\rangle]}_{M^*({\bm{u}})}) {\phi}_{X}(\bm{x}), \\
  =A^* {\phi}_{X}(\bm{x}){+}&\underbrace{\left(\lilsum^{n_u}_{i=1} B^*(\bm{e}_i){\phi}_{X}(\bm{x}) \otimes \bm{e}_i^*\right)}_{B^*_{{\phi}_{X}(\bm{x})}}{\bm{u}},
\end{align}
\end{subequations}
leading to Figure \ref{fig:Equiv}, and concluding the proof.
\end{proof} 
\end{corollary}
  \begin{table*}[htp!]
\scriptsize
% \footnotesize
  \caption{Popular existing operator-theoretic representations for control systems. cKOR includes all (including non-affine) representations with a straightforward modification of sampling/embedding operators.}  \label{tab:cKORpower}
\centering
  \begin{tabular}{l|ccc} \toprule
   \textsc{Method} & \txt{feature dynamics} & \txt{range kernel} $\langle\cdot,\cdot\rangle_{\spOUT}$& \txt{domain kernel} $\langle\cdot,\cdot\rangle_{\spIN}$\\ \midrule
       DMD \cite{schmid_2010} & ${\bm{A}}^*\bm{x}$ & $\innerprod{\bm{x}}{\bm{x}^\prime}$ & $\innerprod{\bm{x}}{\bm{x}^\prime}$\\
    DMDc \cite{Proctor2016} & ${\bm{A}}^*\bm{x}+{\bm{B}}^*\bm{u}$ & $\innerprod{\bm{x}}{\bm{x}^\prime}+\innerprod{\bm{u}}{\bm{u}^\prime}$ & $\innerprod{\bm{x}}{\bm{x}^\prime}$\\
   $k$EDMD \cite{Klus2020eig,Williams_KernelDMD2014} & $\adjoint{A}{\phi}_{X}({\bm{x}})$ & $k_{X}(\bm{x},\bm{x}^\prime)$ & $k_{X}(\bm{x},\bm{x}^\prime)$ \\ 
   $k$EDMDc \cite{Caldarelli2024} & $ \adjoint{A}{\phi}_{X}({\bm{x}})+\adjoint{B}\bm{u}$ & $k_{X}(\bm{x},\bm{x}^\prime)+\innerprod{\bm{u}}{\bm{u}^\prime}$ & $k_{X}(\bm{x},\bm{x}^\prime)$ \\ 
        \bf cKOR & $\adjoint{A}{\phi}_{X}({\bm{x}})+\adjoint{B} {\phi}_{X}({\bm{x}}){\otimes } {\phi}_{U}({\bm{u}})$ & $k_{X}(\bm{x},\bm{x}^\prime)+{k}_{X}(\bm{x},\bm{x}^\prime){k}_{U}({\bm{u}},{\bm{u}^\prime})$ & $k_{X}(\bm{x},\bm{x}^\prime)$\\
    \bottomrule
  \end{tabular}
\end{table*}
While the above result may seem technical, \textit{the established equivalence is of central importance} in practice. Namely, once the isometric isomorphism between two spaces is established, one typically \textit{works with whichever space is more convenient for the problem} at hand, as those shown in Figure \ref{fig:Equiv}. For example, formalizing a regression problem is shown to be cumbersome with input-parameterized operators \cite{Peitz2020,Nuske2023,Philipp2023b}, particularly in infinite dimensions, while it is extremely helpful to build predictors with LPV models in mind. In particular, recasting the operator-based model in an LPV form enables linear-algebraic multi-step prediction and the use of efficient predictive control schemes, e.g., \cite{hoekstra2023computationally}. Conversely, the operator-/vector-valued regression formulation via \eqref{eq:risk} and Theorem \ref{thm:CAkern} enables straightforward and flexible regression, but it may not be immediately apparent how to achieve efficient multi-step prediction and control. As we demonstrate, the established equivalence allows one to use the best of both perspectives: \textit{vector-valued regression for learning and analysis} and the \textit{LPV forms for prediction and control}.
  Not only do our results describe the equivalence between different representations, but they also imply that implicit representations for control operators are completely described through scalar-valued kernels, as summarized in Table \ref{tab:cKORpower}. Hence, our equivalence results help bridge an important gap in understanding linear control operators, enabling the full utilization of the available RKHS structure, both for analysis and learning.
\subsection{Control Koopman operator approximations in RKHS}  
After constructively establishing different equivalent RKHS-based hypotheses, we now study their approximation capabilities. Recall that, under Assumptions \ref{asm:Bcko} and \ref{asm::RKHS}, \textit{the operator restriction is Hilbert-Schmidt} \(\mathcal{T}\defeq{\mathcal{G}}{S}_{\mu^{\prime}}\in\HS{\spIN,L^{2}_{\nu}}\) and approximation of \({\mathcal{G}}\) over functions in \(\spIN\) in \textit{operator norm} is feasible with finite-rank operators. This is critical, as $\|{\cdot}\|_{\mathrm{op}}$, which measures the worst-case difference between two operators in a target (image) space, leads to bounds quantifying the maximum possible error effect on any observable from the domain $\spIN$. 
Searching over $\operatorname{HS}(\spIN,\spOUT)$ via reproducing kernels and interpreting the image space as $\truOUT$,  \eqref{eq:MSEriskmin} yields a principled surrogate for infinite-dimensional regression \cite{Mollenhauer2020}, so that \eqref{eq:HSrisk} sharply bounds the $\|{\cdot}\|_{\spIN  {\rightarrow}  L^{2}_{\nu}}^{2}$-error.
\begin{lemma}\label{lem:riskBound}
        Under Assumptions \ref{asm:Bcko} and \ref{asm::RKHS}, \eqref{eq:risk} sharply satisfies $\|{\mathcal{G}}-G\|_{\spIN  {\rightarrow}  L^{2}_{\nu}}^{2} \leq \mathsfit{R}(G)$ for every \(G \in \HS{\spIN,\spOUT}\).
\begin{proof}
 Let \(G \in \HS{\spIN,\spOUT}\). We have that 
\begin{subequations}
    \begin{align}
&\|{\mathcal{G}}-G\|_{\spIN  {\rightarrow}  L^{2}_{\nu}}^{2} {=}\textstyle\sup _{\|f\|_{\spIN}=1}\|\mathcal{T} f- G f\|_{L^{2}_{\nu}}^{2}\\
& {=}\textstyle\sup _{\|f\|_{\spIN}=1}\|[\mathcal{T} f](\cdot) - [G f](\cdot)\|_{L^{2}_{\nu}}^{2} \label{eq:preRP}\\
& {=}\sup _{\|f\|_{\spIN}=1}\|\left\langle f, \phi_X(\bm{\mathsf{f}}(\cdot))\right\rangle_{\spIN}-\left\langle  f,  \adjoint{G} \phi_{Z}(\cdot)\right\rangle_{\spIN}\|_{L^{2}_{\nu}}^{2} \label{eq:postRP} \\
& {=}\textstyle\sup _{\|f\|_{\spIN}=1} \mathbb{E}_{\bm{z} \sim \nu}\left[\left\langle f,\phi_X(\bm{\mathsf{f}}(\bm{z}))- G^{*} \phi_{Z}({\bm{z}})\right\rangle_{\spIN}^{2}\right] \label{eq:preCaSch}\\
& \leq \sup_{\|f\|_{\spIN}=1} \!\!\! \mathbb{E}_{\bm{z} \sim \nu}\left[\|f\|_{\spIN}^{2}\left\|\phi_X(\bm{\mathsf{f}}(\bm{z})){-}G^{*} \phi_{Z}({\bm{z}})\right\|_{\spIN}^{2}\right]\label{eq:CaSch}
\end{align}
\end{subequations}
where we use the definition \eqref{eq:cKOp} and the reproducing property in \eqref{eq:preRP}-\eqref{eq:postRP} together with the Cauchy-Schwarz inequality in \eqref{eq:CaSch}. We say the upper bound is sharp if it holds for all $f {\in} \spIN$ and $ \exists h {\in} \spIN$ so $ \mathsfit{R}(G){=}  \|{\mathcal{G}}{-}G\|_{\spIN {\rightarrow}  L^{2}_{\nu}}^{2}$. Hence, the above bound is sharp by considering that we have $\nu$-a.e. \(\phi_X(\bm{\mathsf{f}}(\bm{z})){-}G^{*} \phi_{Z}({\bm{z}}){=}e\) for some constant \(e \in \spIN\) so the equality is attained by setting \(f{=}e /\|e\|_{\spIN}\) in the supremum.
\end{proof}
\end{lemma}
As the following theorem shows, with the help of Lemma \ref{lem:riskBound}, the operator norm error vanishes for $C_0$-universal $\RKHS_X, \RKHS_U$.
\begin{theorem}[Arbitrary accuracy]\label{thm:opNorm}
Let Assumptions \ref{asm:Bcko} and \ref{asm::RKHS} hold and let $\spOUT \supseteq \RKHS_U \otimes \RKHS_X$ be an RKHS. Then
    \begin{enumerate}
        \item for every $\delta>0$, there is a \textit{finite-rank} $G \in \HS{\spIN,\spOUT}$ such that
        \[
            \underbrace{\mathsfit{E}(G)\defeq\|{\mathcal{G}}-G\|_{\spIN {\rightarrow} L^{2}_{\nu}}}_{\text{operator norm error}}<\underbrace{\|[I{-}P_{\spOUT}]{\mathcal{G}}\|_{\spIN {\rightarrow} L^{2}_{\nu}}}_{{\text{representation bias}~\mathsfit{B}(\spOUT)}}~+~\delta,
        \]
        with $P_{\spOUT}$ the orthogonal projector onto $\mathrm{cl}(\mathrm{Im}\,S_\nu) \subseteq L^2_\nu(\Set{Z})$;
       \item for every $\delta>0$, there is a \textit{finite-rank} $G \in \HS{\spIN,\spOUT}$ such that $\mathsfit{E}(G)<\delta$, if and only if $\mathsfit{B}(\spOUT)=0$; in particular, this holds when $\spIN$ and $\RKHS_U$ \textit{are} $C_0$-\textit{universal}.
    \end{enumerate}
\end{theorem}
\begin{proof}
Let $\delta > 0$. By the triangle inequality
\begin{align*}
  \mathsfit{E}(G)\leq&~\|[I-P_{\spOUT}]{\mathcal{G}}\|_{\spIN {\rightarrow} L^{2}_{\nu}} + \|P_{\spOUT}{\mathcal{G}}-G\|_{\spIN {\rightarrow} L^{2}_{\nu}} \\
  \leq&~\mathsfit{B}(\spOUT) + \|S_\nu\|\|P_{\spOUT}{\mathcal{G}}-G\|_{\spIN \rightarrow \spOUT} 
\end{align*}
where the error splits into \textit{representation bias} $\mathsfit{B}(\spOUT)$ and \textit{rank reduction error} $\|P_{\spOUT}{\mathcal{G}}-G\|_{\spIN {\rightarrow} \spOUT}$. By the fact that finite-rank operators from $\spIN \to \spOUT$ are dense in $\HS{\spIN,\spOUT}$, we have that $\|S_\nu\|\|P_{\spOUT}{\mathcal{G}}-G\|_{\spIN \rightarrow \spOUT} < \delta$ so that $\mathsfit{E}(G) < \mathsfit{B}(\spOUT) +\delta$. 
\newline
\textit{If $\spIN, \RKHS_U$ are $C_0$-universal}: $\RKHS_U \otimes \RKHS_X$ is dense in $L^2_{\nu}(\Set{Z})$, hence $\operatorname{Im}(\mathcal{T}) \subseteq \operatorname{cl}(\operatorname{Im}(S_\nu))$ for $\mathcal{T}\defeq\mathcal{G}{S}_{\mu^\prime}$, so that $\|[I-P_{\spOUT}]{\mathcal{G}}\|_{\spIN {\rightarrow} L^{2}_{\nu}}=\mathsfit{B}(\spOUT)=0$ \cite[Chapter 4]{IngoSteinwart2008SupportMachines}.
\newline
\textit{If $\mathsfit{B}(\spOUT)=0$}: Let ${T}^\dagger(\cdot)\defeq {\phi}_{X}(\bm{\mathsf{f}}(\cdot))$ and let $\mathscr{T}$ be the vector-valued RKHS isometrically isomorphic to $\HS{\spOUT,\spIN}\cong\spIN \otimes \spOUT$ via \cite[Cor.~4.5]{Mollenhauer2020}:
\begin{equation}\label{eq:Tspace}
    \!\!\mathscr{T}\defeq\{{T}\!:\Set{Z} \to \spIN \mid T  {=}  G^{*} \phi_{Z}(\cdot), G^*{\in} \HS{\spOUT,\spIN}\},
\end{equation}
induced by the operator-valued kernel $K_{Z}(\bm{z},\bm{z}^\prime) \defeq {k_Z}(\bm{z},\bm{z}^\prime)I_{\spIN}$. Since $\mathsfit{B}(\spOUT)=0$ gives $\operatorname{Im}(\mathcal{T}) \subseteq \operatorname{cl}(\operatorname{Im}(S_\nu))$, the target $T^\dagger$ lies in the closure of $\mathscr{T}$ \eqref{eq:Tspace} inside $L^{2}_{\nu}(\spIN)\equiv L^{2}_{\nu}(\Set{Z},\spIN) \cong \HS{\spIN, L^{2}_{\nu}(\Set{Z})}$; hence, for every $\delta>0$ there exists $G^{*} \in \HS{\spOUT,\spIN}$ with corresponding $T\defeq G^*\phi_{Z}(\cdot)\in\mathscr{T}$ satisfying $\|T^\dagger-T\|_{L^{2}_{\nu}(\spIN)}^{2} = \|\mathcal{T}-S_\nu G\|_{\mathrm{HS}}^2 \equiv \mathsfit{R}(G) < \delta^2$. Lemma \ref{lem:riskBound} bridges this HS bound to the operator norm: $\mathsfit{E}^2(G) \leq \mathsfit{R}(G) < \delta^2$, so $\mathsfit{E}(G)<\delta$, completing the proof.\end{proof}
The above result reveals that whenever the RKHSs $\spIN,\RKHS_U$ used to define \eqref{eq:CAkern} are $C_0$-universal, then there is no representation bias/error $\mathsfit{B}(\spOUT) = 0$ and one can find arbitrarily
good finite-rank approximations of control Koopman operators \ref{eq:cKOp}. Note that Assumption \ref{asm::RKHS} on the RKHSs $\spIN,\spOUT$ is non-restrictive and not actually an assumption on the problem -- it solely depends on the choice of the kernel and is readily satisfied by popular kernels such as Gaussian, Laplacian or Matérn kernels \cite{IngoSteinwart2008SupportMachines}.
Moreover, notice that we do not require the true operator ${\mathcal{G}}\!: L^2_{{{\mu^{\prime}}}} \to L^2_{\nu}$ to be compact, let alone Hilbert-Schmidt, for Theorem \ref{thm:opNorm} to hold. 
\subsection{Approximating the dynamics of observables}
One of the practical appeals of control operator models is the \textit{ability to forecast of any observable belonging to the hypothetical domain} $\spIN$. For that, operator norm approximation is critical, allowing arbitrarily accurate prediction of any measurement/observable $y \in \spIN$. We formalize this in the following corollary.
\begin{corollary}\label{coro:Opred}
Let the conditions of Theorem \ref{thm:opNorm} hold for $C_0$-universal $\RKHS_X, \RKHS_U$, and denote the true one-step evolution of an observable $y \in \spIN$ as ${y}_+(\bm{z})\defeq [{\mathcal{G}}{S}_{\mu^{\prime}} y](\bm{z})$. Then, for every $y \in \spIN$ and $\varepsilon>0$, there exists a $G \in \HS{\spIN,\spOUT}$ such that ${\|{y}_+-{S}_{\nu}Gy\|_{L^2_{\nu}}} < \varepsilon$.

\begin{proof}
   The operator-norm bound gives ${\|({\mathcal{G}}{S}_{\mu^{\prime}} -{S}_{\nu}G)y\|_{L^2_{\nu}}}\leq \mathsfit{E}(G)\|y\|_{\spIN}$, and Theorem~\ref{thm:opNorm} with $\delta = \nicefrac{\varepsilon}{\|y\|_{\spIN}}$ yields the assertion ($y=0$ immediate).
\end{proof}
\end{corollary} 
   The main added assumption in the above result is $y \in \spIN$, which may be straightforwardly fulfilled for known observables of interest, e.g., components of the full-state observable, by adding a linear kernel component to the hypothetical domain to include the state: ${\RKHS}_X^{\txt{id}} = \RKHS_X{\oplus}\txt{Id}(\Set{X})$, which is induced by the kernel ${k}^{\txt{id}}_{X}(\bm{x}, \bm{x}^\prime) = k_{X}(\bm{x},\bm{x}^\prime)+\innerprod{\bm{x}}{\bm{x}^\prime}$. Still, such a composite RKHS is universal if $k_{X}(\bm{x},\bm{x}^\prime)$ is $C_0$-universal, allowing Theorem \ref{thm:opNorm} to hold with a vanishing representation bias ${\mathsfit{B}(\cdot)}= 0$. However, existing approaches append the state to a data-independent and finite dictionaries \cite{KORDA2018149,STRASSER20232257} or evaluate an empirical inner product that is only an approximation of the canonical ${k}^{\txt{id}}_{X}(\bm{x}, \bm{x}^\prime)$. As such, there is no guarantee for unbiased representations ${\mathsfit{B}(\cdot)}= 0$ (cf. Theorem \ref{thm:opNorm}), which is crucial for achieving statistical consistency \cite{IngoSteinwart2008SupportMachines}.
   Note that we do not require \eqref{eq:cKOp} to be compact (let alone Hilbert-Schmidt), while our hypothesis does not need to be invariant w.r.t. \eqref{eq:cKOp} to admit arbitrarily accurate prediction.
\begin{remark}[Improved analysis]
    In contrast to our nonparametric and discrete-time setting, existing operator approximation error analyses for control use pre-RKHS hypothesis approaches \cite{Nuske2023}, even when using kernels \cite{Philipp2023b}, and do not directly benefit from the structure of an RKHS. Additionally, the analysis therein comes with an irreducible time-discretization error and explicitly depends on control input dimensionality, e.g., using infinite-dimensional control spaces (cf. Remark \ref{rmk:infU}) would render regression and analysis infeasible for the approaches in \cite{Peitz2020,Nuske2023,Philipp2023b}. Although the aforementioned works do focus on probabilistic finite-data error bounds for a set of constant input Koopman operators, our theoretical analysis indicates that sharper and more flexible results should be readily available. Exploring this further is, however, out of the scope of this work.
\end{remark}
\section{Data-Driven Control Koopman Operators% over RKHSs
}\label{sec:cKOR}
\subsubsection*{Regularized risk minimization} After establishing the Hilbert-Schmidt (HS) representations over RKHSs, we continue with formulating a well-posed regression problem to solve \eqref{eq:MSEriskmin}. Recognizing that the estimator-dependent, projected component of \eqref{eq:HSrisk} admits a unique minimizer 
\begin{align}\label{eq:minNorm}
        \!\!\!\!\!\argmin_{\substack{\innerprod{ S_{\nu} g}{S_{\nu}G f} = \innerprod{S_{\nu}g}{\mathcal{T} f}\\ f \in \spIN, g \in \spOUT }}\|{G} \|_{\HS{\spIN,\spOUT}} = \left({C}_{ZZ}\right)^{\dagger} C_{ZX_+},
\end{align}
via the cross-covariance $C_{ZX_+} \defeq S^*_{\nu}\mathcal{T}$ and covariance $C_{ZZ} \defeq S^*_{\nu}S_{\nu}$, where $\dagger$ denotes the Moore-Penrose pseudoinverse operator. However, the latter may still lead to a poorly-conditioned system of equations. To ensure well-posedness, we use Tikhonov regularization with $\gamma>0$
 \begin{align}\label{eq:infKRR}
\infEstim_{{\reg}} 
   & \defeq \argmin_{\Estim \in \HS{\spIN,\spOUT}}\mathsfit{R}(G)+ {\reg}\norm{\Estim}^{2}_{{\rm HS}}\notag\\ &= (C_{ZZ}+{\reg} {\mathrm{Id}_{\spOUT}})^{-1}C_{ZX_+} . \tag{\txt{r\textsc{rm}}}
\end{align}
It is easy to confirm that the objective in \eqref{eq:infKRR} is continuous, coercive, and strictly convex, making \eqref{eq:infKRR} a \textit{unique} minimizer \cite{Engl2015} of the regularized risk, incurring no loss of precision as existing finite-dimensional finite-section approaches \cite{Philipp2023b,morrison2023dynamic}.
\preSectionHard
\subsection{Estimating control Koopman operators from data}
Computing \eqref{eq:infKRR} would require access to all values of \eqref{eq:ncs} under the population distribution $\nu$ on $\Set{Z}$. Since this distribution is unavailable when learning from data, we rely on finite samples 
\begin{align}
    \Set{D}_{n}= \bigl\{\bm{z}^{(i)}\defeq (\bm{x}^{(i)},\bm{u}^{(i)}),\bm{x}_{+}^{(i)} \bigr\}_{i{\in}[n]}
\end{align}
that define \emph{sampling} operators
\begin{subequations}\label{eq:sampl}
    \begin{align}
\ES_Z:\spOUT\to\Set{R}^n,\quad \ES_{Z}h &= [h(\bm{z}^{(1)}) \cdots h(\bm{z}^{(n)}) ]^\top, \\
 \widehat{S}_U:\RKHS_U\to\mathbb{R}^n, \quad \ES_{U}g & =[g(\bm{u}^{(1)}) \cdots g(\bm{u}^{(n)}) ]^\top,\\
\ES_X:\spIN\to\Set{R}^n,\quad\ES_{X}f &= [f(\bm{x}^{(1)}) \cdots f(\bm{x}^{(n)}) ]^\top,
\\
\ES_+:\spIN\to\Set{R}^n,\quad \ES_{+}f &= [f(\bm{x}^{(1)}_+) \cdots f(\bm{x}^{(n)}_+) ]^\top. \label{eq:sample-succ}
\end{align} 
\end{subequations}
so their adjoints, the sampled \textit{embedding} operators \cite{smale2007learning}, are defined as $\aES_{Z}\bm{a} = \lilsum^n_{i=1}\phi_{Z}({\bm{z}^{(i)}}) (\bm{a})_i$, $\adjoint{\ES}_U\bm{b} = \lilsum^n_{i=1}{\phi}_{U}({\bm{u}^{(i)}})(\bm{b})_i$, $\aES_X\bm{c} = \lilsum^n_{i=1}{\phi}_{X}({\bm{x}^{(i)}}) (\bm{c})_i$ and $\aES_+\bm{d} = \lilsum^n_{i=1}{\phi}_{X}({\bm{x}_+^{(i)}}) (\bm{d})_i$. Using sampling operators as estimates for \eqref{eq:incl}, we obtain the regularized empirical risk
\begin{equation}\label{eq:Erisk}
    \widehat{\mathsfit{R}}_\gamma(G)\defeq {\textstyle\frac{1}{n}}\|\ES_+{-}\ES_Z G\|^2_{\HS{\spIN,\Set{R}^n}}+\gamma\|G\|^2_{\mathrm{HS}}\tag{$\widehat{\txt{\textsc{risk}}}$}
\end{equation}
and risk minimization  
\begin{align}\label{eq:KRR}
\!\!\!\EEstim_{{\reg}} 
\defeq\!\!\argmin_{{\Estim \in \HS{\spIN,\spOUT}}}{\!\!\!\!\!\!\widehat{\mathsfit{R}}_\gamma(G)}=\big(\aES_Z \ES_Z{+}n\gamma {\mathrm{Id}_{\spOUT}}\big)^{-1}\!\aES_Z \ES_+\,
\tag{$\widehat{\txt{r\textsc{rm}}}$} 
\end{align}
called \textit{kernel ridge regression} ({KRR}) \cite{IngoSteinwart2008SupportMachines,muandet2017kernel,kostic2022learning}. Other popular estimators from dynamical systems literature, such as \textit{principal component} (PCR) \cite{Williams2016,Williams_KernelDMD2014} and \textit{reduced rank} (RRR) regression \cite{kostic2022learning} are compatible with our setting by replacing the matrix ${\Kreg^{-1}}$ in Algorithm \ref{alg:ckor}..
\subsubsection*{Finite-dimensional predictors} When both $\spIN$ and $\spOUT$ are infinite-dimensional universal RKHSs, we can not instantiate the estimate \eqref{eq:KRR} on computer hardware directly as representer theorems \cite{Micchello2006,Carmeli2010VectorUniversality} only ``collapse" $\spOUT$, due to $\mathrm{dim(range}(\widehat{G}_\gamma)) < \infty$, making \eqref{eq:KRR} tractable when $\mathrm{dim}(\spIN) < \infty$. Still, as formalized in the following, \textit{prediction using $\widehat{G}_\gamma$ requires only matrix-vector multiplication.}
\begin{algorithm}[t]
\caption{\textsc{cKOR Model Learning \& Inference}}
\label{alg:ckor}
\begin{algorithmic}
\algoSepIn
 \Statex {\bfseries Input:} Data $\{\bm{x}^{(i)},\bm{u}^{(i)},\bm{x}^{(i)}_{+}\}_{i=1}^n$, state kernel $k_{X}(\bm{x},\bm{x}^\prime)\defeq  \innerprod{{\phi}_{X}(\bm{x})}{{\phi}_{X}(\bm{x}^\prime)}$ with  ${\phi}_{X}: \Set{X} \to \RKHS_X$, control kernel $k_{U}(\bm{u},\bm{u}^\prime)\defeq  \innerprod{{\phi}_{U}(\bm{u})}{{\phi}_{U}(\bm{u}^\prime)}$ with ${\phi}_{U}: \Set{U} \to \RKHS_U$, observable $\bm{y} \in \RKHS_X^{n_y}$, regularization $\gamma>0$.
  \algoSep
  \Statex{\texttt{/* Nonparametric Representation */}}
    \Statex \textbf{let:} $\bm{I}\defeq \mathrm{diag}(\bm{1})$ with $\bm{1}\in\mathbb{R}^n$,\; $(\bm{1})_i = 1$
  \Statex \textbf{let:} $\bm{k}_{X}(\bm{x}) \in \mathbb{R}^n$, \; $(\bm{k}_{X}(\bm{x}))_i = k_{X}(\bm{x},\bm{x}^{(i)})$
  \Statex \textbf{let:} $\bm{k}_{U}(\bm{u}) \in \mathbb{R}^n$, \; $(\bm{k}_{U}(\bm{u}))_i = k_{U}(\bm{u},\bm{u}^{(i)})$
  \Statex \textbf{let:} $\bm{K}_{X} \in \mathbb{R}^{n\times n}$, \; $(\bm{K}_{X})_{ij} = k_{X}(\bm{x}^{(i)},\bm{x}^{(j)})$
  \Statex \textbf{let:} $\bm{K}_{U} \in \mathbb{R}^{n\times n}$, \; $(\bm{K}_{U})_{ij} = k_{U}(\bm{u}^{(i)},\bm{u}^{(j)})$
    \Statex \textbf{let:} $\bm{K}_{+} \in \mathbb{R}^{n\times n}$, \; $(\bm{K}_{+})_{ij} = k_{X}(\bm{x}_+^{(i)},\bm{x}^{(j)})$
  \algoSep
  \Statex{\texttt{/* Model Learning */}}
  \Statex \textbf{compute:} $\bm{K}_Z = \bm{K}_{X} \odot (\bm{1}\bm{1}^\top+\bm{K}_{U})$
  \Statex \textbf{compute:} ${\Kreg^{-1}} = (\bm{K}_Z + n\gamma\,\bm{I})^{-1}$
  \Statex \textbf{compute:} $\bm{A} = ({\Kreg^{-1}}\bm{K}_{+})^\top$ \label{ssm:Begin}
  \Statex \textbf{let:} $\mathsfit{B}\bm{\mathsf{v}} \defeq  \mathrm{diag}(\bm{\mathsf{v}})\bm{A} \in \Set{R}^{n \times n}$ 
  \Statex \textbf{let:} $\bm{\mathsf{z}}(\bm{x},\bm{u}) \defeq \bm{k}_{X}(\bm{x})\odot(\bm{1}+\bm{k}_{U}(\bm{u})) \in \Set{R}^n$ \label{ssm:End}
  \algoSep
\Statex{\texttt{/* Inference */}}
\Statex \textbf{let:} ${\bm{Y}}_+=[\bm{y}({\bm{x}}_+^{(1)}),\dots,\bm{y}({\bm{x}}_+^{(n)})]^\top$
  \Statex \textbf{compute:} $\bm{C} = ({\Kreg^{-1}}\bm{Y}_+)^\top$ 
  \algoSep
  \Statex {\bfseries Return:} matrices $(\bm{A},\bm{C})$, map $\mathsfit{B}$, lifting $\bm{\mathsf{z}}(\cdot,\cdot)$.
\end{algorithmic}
\end{algorithm}
\begin{algorithm}[t]
\caption{\textsc{Nonparametric State-Space Prediction}}
\label{alg:ckorpred}
\begin{algorithmic}
\algoSepIn
 \Statex {\bfseries Input:} Matrices $(\bm{A},\bm{C})$, matrix-valued map $\mathsfit{B}$, lifting $\bm{\mathsf{z}}(\cdot,\cdot)$, initial condition $\bm{x}_0$ and input sequence $\{\bm{u}_k\}^{H-1}_{k=0}$.
  \algoSep
  \Statex{\texttt{/* Finite-Horizon Prediction */}}
  \Statex $\bm{\mathsf{z}}_1 = \bm{\mathsf{z}}(\bm{x}_0,\bm{u}_0)$ 
  \For{$k=1,2,\dots,H-1$}
        \State $\widehat{\bm{y}}_{k} = \bm{C}\,\bm{\mathsf{z}}_{k}$
        \State $\bm{\mathsf{z}}_{k+1} = \big(\bm{A}+\mathsfit{B}\,\bm{k}_{U}(\bm{u}_k)\big)\,\bm{\mathsf{z}}_{k}$
  \EndFor
  \Statex $\widehat{\bm{y}}_{H} = \bm{C}\,\bm{\mathsf{z}}_{H}$
  \algoSep
  \Statex {\bfseries Return:} predictions $\{\widehat{\bm{y}}_k\}_{k=1}^H$.
\end{algorithmic}
\end{algorithm}
\begin{proposition}\label{prop:cKORkrrSHORT} Consider control $\RKHS_U$ and state $\RKHS_X$ RKHSs, forming the domain $\RKHS_X$ and range $\RKHS_Z \defeq \spIN \oplus (\RKHS_U \otimes \RKHS_X)$ for $G$ in \eqref{eq:MSEriskmin}. Given an observable of interest $\bm{y} \in \RKHS^{n_y}_X$, the repeated application of \eqref{eq:KRR} is equivalent to a finite-dimensional state-space model defined by Algorithm \ref{alg:ckor} \& \ref{alg:ckorpred}.

\begin{proof}
By the \emph{push-through identity} and the reproducing property, we rewrite \eqref{eq:KRR} as $\aES_Z(\ES_Z \aES_Z+n\gamma \bm{I})^{-1} \ES_+ = \aES_Z{\Kreg^{-1}}\ES_+$. Then, by the structure of $\spOUT$, we have
\begin{subequations}\label{eq:KRRmain}
        \begin{align}
        \EAEstim_{{\reg}}&= \aES_X{\Kreg^{-1}}\ES_+\in \HS{\spIN,\spOUTx},\\
        \EBEstim_{{\reg}}&= {(\widehat{S}^*_U\circledcirc\ES^*_X)}{\Kreg^{-1}}\ES_+ \in \HS{\spIN,\RKHS_U\otimes \spOUTx},
    \end{align}
\end{subequations}
where we used the Khatri–Rao product \(\widehat S_U^{*}\circledcirc \widehat S_X^{*}\), defined by \((\widehat S_U^{*}\circledcirc \widehat S_X^{*})e_i=(\widehat S_U^{*}e_i)\otimes(\widehat S_X^{*}e_i)\) for an orthonormal basis \(\{e_i\}\); so that \( (\widehat S_U^{*}\circledcirc \widehat S_X^{*})^{*}(v\otimes w)=(\widehat S_U v)\odot(\widehat S_X w)\).
Moreover, we have that $\bm{K}_Z \defeq \ES_Z \aES_Z = {\bm{K}_{X}}+\bm{K}_{U}{\odot} {\bm{K}_{X}}\in \Set{R}^{n \times n}$ after using the identity $(\widehat{S}_U\circledcirc\ES_X)\adjoint{(\widehat{S}_U\circledcirc\ES_X)}=(\widehat{S}_U\adjoint{\widehat{S}_U}){\odot} (\ES_X\aES_X)$.
 For $\widehat{G}^*_\gamma=\aES_+{\Kreg^{-1}}{\ES_Z}$, or equivalently, $\adjoint{\EAEstim}_{{\reg}},\adjoint{\EBEstim}_{{\reg}}$, we can identify the weighted forward embedding $\adjoint{\widehat{F}}=\aES_+{\Kreg^{-1}}$ so that ${\bm{{A}}}=\ES_X\adjoint{\widehat{F}} \in \Set{R}^{n \times n}$. Propagating a scalar-valued observable $y \in \spIN$ for one step ($k=1$) gives
\begin{subequations}
    \begin{align}
    \!\!\!\!\!\widehat{{y}}_{1}(\bm{z}_0)\defeq~&[\EEstim_\gamma y](\bm{z}_0) = \innerprod{{{y}}}{\adjoint{\EEstim}_\gamma \phi_{Z}({\bm{z}_0})}_{\spIN}=\\
    =~&\innerprod{{y}}{\underbrace{\widehat{A}_\gamma^*{\phi}_{X}(\bm{x}_0){+}\widehat{B}_\gamma^*({\phi}_{U}({\bm{u}_0}) {\otimes}  {\phi}_{X}(\bm{x}_0))}_{ \mathsfit{z}_{1}({\bm{x}_{0}},\bm{u}_0)}}_{\spIN}.
\end{align}
\end{subequations}
By substituting operators \eqref{eq:KRRmain} and $\adjoint{\widehat{F}}$ we have that $\mathsfit{z}_{1}({\bm{x}_{0}},\bm{u}_0)$ equals to
\begin{subequations}
    \begin{align}
&\adjoint{\widehat{F}}( \widehat{S}_X\phi_{X}({\bm{x}_{0}})+{\widehat{S}_U{\phi}_U({\bm{u}_0})}{\odot} \ES_X {\phi}_{X}({\bm{x}_{0}})), \\
   &= \adjoint{\widehat{F}}( \bm{k}_{X}({\bm{x}_{0}})+{ \bm{k}_{U}({\bm{u}_{0}})}{\odot} \bm{k}_{X}({\bm{x}_{0}})), \\
  &= \adjoint{\widehat{F}} (\bm{k}_{X}({\bm{x}_{0}}) \odot (\bm{1}+\bm{k}_{U}(\bm{u}_0)))\defeq \adjoint{\widehat{F}}{\bm{\mathsf{z}}_1}.
\end{align} 
\end{subequations}
For $k=2$, we plug in the previous solution for $\mathsfit{z}_{1}({\bm{x}_{0}},{\bm{u}_{0}})$
\begin{subequations}
    \begin{align}
    \mathsfit{z}_{2}({\bm{x}_{0}},{\bm{u}_{0}},{\bm{u}_{1}}) &= \adjoint{\widehat{F}}({\bm{{A}}}{{{{\bm{\mathsf{z}}}}}_1}+{\widehat{S}_U{\phi}_U({\bm{u}_1})}{\odot} {\bm{{A}}}{{{{\bm{\mathsf{z}}}}}_1}) \\
        &= \adjoint{\widehat{F}}({\bm{{A}}}{{{{\bm{\mathsf{z}}}}}_1}{+}\mathrm{diag}({\bm{k}_{U}{(\bm{u}_{1})}}) {\bm{{A}}}{{{{\bm{\mathsf{z}}}}}_1}) \\
                &\defeq\adjoint{\widehat{F}}(\bm{A}+\mathsfit{B}\,\bm{k}_{U}(\bm{u}_{1})){\bm{\mathsf{z}}_1}.
\end{align}
\end{subequations}
It is straightforward to verify by induction that 
\[
   \widehat{{y}}_{H}(\bm{x}_0,\{\bm{u}_k\}^{H-1}_{k=0})=
   \begin{cases}
      \innerprod{\widehat{F}{y}}{\bm{\mathsf{z}}_1}, &\!\!\!H{=}1, \\
      \innerprod{\widehat{F}{y}}{\textstyle\prod_{k=1}^{H-1}\!\!(\bm{A}{+}\mathsfit{B}\bm{k}_{U}(\bm{u}_{k}))\bm{\mathsf{z}}_1}, & \!\!\!H{>}1.
   \end{cases}
\]
yields the predictor from Algorithm~\ref{alg:ckorpred} based on Algorithm~\ref{alg:ckor} after using the reproducing property column-wise $\widehat{\bm{y}}_{1}=[\,\widehat{F} y_1 \;\cdots\;\widehat{F} y_{n_y}\,]  = \Kreg^{-1} \bm{Y}_+ \defeq \bm{C}^\top$, accommodating vector-valued observables.
\end{proof}
\end{proposition}
 The above result shows that, after a sequence of control inputs the initial condition is transformed to $\{\bm{k}_{U}({\bm{u}_k})\}^{H-1}_{k=0},\bm{k}_{X}({\bm{x}_0})$, \textit{the prediction involves only matrix-vector multiplication}.
    While the above predictor is finite-dimensional, it is the exact representation of the infinite-dimensional finite-rank control Koopman operator \eqref{eq:KRR}. In turn, Proposition \ref{prop:cKORkrrSHORT} eliminates auxiliary regression problems, e.g., for state reconstruction, as the full-state is readily reconstructed by setting $\bm{y} = {\mathrm{id}}$. 
\begin{remark}[Bilinear state--space models]
If the control kernel is linear, $k_{U}(\bm{u},\bm{u}')=\bm{u}^\top \bm{u}'$, then 
$\bm{k}_{U}(\bm{u}) = U\bm{u}$ with $U=[\bm{u}^{(1)},\dots,\bm{u}^{(n)}]^\top$, and $\mathsfit{B}\,\bm{k}_{U}(\bm{u})
= \mathrm{diag}(U\bm{u})\bm{A}$. The prediction recursion from Algorithm \ref{alg:ckor} takes the form
% \[
\begin{align}\label{eq:BilinMOD}
    \bm{\mathsf{z}}_{k+1} = \bm{A}\bm{\mathsf{z}}_k + \textstyle \sum_{i=1}^{n_u} u_{k,i}\,\bm{B}_i \bm{\mathsf{z}}_k,
\end{align}
% \]
which is a bilinear state--space model with control channels $\{\bm{B}_i \defeq \mathrm{diag}(U\bm{e}_i)\bm{A}\in \Set{R}^{n \times n}\}^{n_u}_{i=1}$. For a general nonlinear kernel $k_{U}$, the same recursion 
$\bm{\mathsf{z}}_{k+1} = \big(\bm{A}+\mathsfit{B}\,\bm{k}_{U}(\bm{u}_k)\big)\bm{\mathsf{z}}_k$
is now bilinear in $(\bm{k}_{U}(\bm{u}),\bm{\mathsf{z}})$, presenting a kernelized generalization of bilinear state--space models.
\end{remark} 
\subsubsection*{Embracing infinite-dimensional regression}
    Although the finite sample solution of operator regression over infinite-dimensional RKHSs is itself infinite-dimensional, it is known
to be of \emph{finite rank}
\cite{muandet2017kernel,Mollenhauer2020,Mollenhauer2022,kostic2022learning}. This makes for a finite-dimensional range $\mathrm{dim(range}(\widehat{G}_\gamma)) < \infty$, but leaves the domain $\spIN$ intact and infinite-dimensional.
However, kernelized operator methods for control systems \cite{rosenfeld2024dynamic,morrison2023dynamic,Philipp2023b} discretize the domain by replacing the canonical inner product
$\langle\cdot,\cdot\rangle_{\mathcal H_X}$ with the inner product\footnote{%
With the identification
\(
  \mathcal{H}_{X}^{0}
  \equiv
  \bigl\{\sum_{i=1}^{n}c_i\,k_{X}(\,\cdot\,,\bm{x}_i):\bm c\in\mathbb R^{n}\bigr\},
\)
we have the inner product
\(
  \langle f,g\rangle_{\mathcal{H}_{X}^{0}}
  =
  \bm{f}^{\intercal}\bm{K}\bm{g},
\)
where $\bm{K}_{ij}=k_{X}(\bm{x}_i,\bm{x}_j)$ and
$f=\sum_i f_i k_{X}(\,\cdot\,,\bm{x}_i),\; g=\sum_i g_i k_{X}(\,\cdot\,,\bm{x}_i)$
\cite{IngoSteinwart2008SupportMachines,
      kanagawa2025gaussianprocessesreproducingkernels}.} induced on
a finite-dimensional \emph{pre‑RKHS} $\mathcal{H}^{0}_{X}\defeq\operatorname{span}\{\phi_X(\bm{x}_i)\equiv k_X(\,\cdot\,,\bm{x}_i)\}_{i=1}^n$, i.e. the span of kernel sections evaluated at the data points. The latter, in turn, discretizes the original problem, e.g., \eqref{eq:KRR}, and leads to the finite-dimensional optimization
        \begin{align}\label{eq:findimPRE}
      \min_{\substack{\innerprod{f}{Gg}=\mathsfit{T}(\innerprod{f}{g}) \\ f,g \in {\RKHS^0_X}}}\|{{G}}\|_{\mathrm{HS}} \longleftrightarrow 
      \min_{\bm{K}\bm{G} = \bm{T} }\|{\bm{G}}\|_{\mathrm{F}},
    \end{align}
    where $\mathsfit{T}(\cdot)$ encodes the action of the operator, $\bm{A},\bm{T}\in \mathbb{R}^{n\times n}$ are computed using kernel function evaluations and $\|\cdot\|_{F}$ denotes the Frobenius norm.
While~\eqref{eq:findimPRE} indeed produces a solution of finite rank, this is only because the search space itself is finite-dimensional. Such a formulation bypasses the inherent infinite-dimensional nature of the original regression problem~\eqref{eq:KRR} and consequently introduces an additional discretization error. Therefore,  the methods of \cite{rosenfeld2024dynamic,morrison2023dynamic,Philipp2023b} do not enjoy the approximation guarantees that we establish under minimal assumptions (Theorem~\ref{thm:opNorm} and Corollary~\ref{coro:Opred}).

\section{Efficient Approximations via Sketching}\label{sec:Sketch}
Like any other nonparametric approach, our cKOR algorithm is only suitable for small-scale systems due to the computational time-complexity of order $\mathcal{O}(n^3)$ w.r.t. the data cardinality $n$.
For approximating kernel methods, random Fourier features (RFFs) stand out as a popular and straightforward way to reduce the time-complexity of estimation \cite{Rahimi2007}. Recently, they have been utilized for control system identification \cite{Philipp2023b}. Unfortunately, the algorithm in \cite{Philipp2023b} is hardly useful in data-driven applications as it requires data gathered under constant inputs to estimate the Koopman operator for a few fixed input levels -- prohibiting most realistic system identification scenarios that involve rich excitation signals or safe data collection, e.g., under an auxiliary controller. Moreover, RFFs being data-independent, they may not adapt well to the data at hand, limiting performance for an equivalent complexity as sketching schemes \cite{pmlr-v162-chatalic22a}. 
Random sketching or Nystr\"{o}m approximations estimate the kernel matrix by selecting a small subset of ${{m}}$ data points known as inducing points or Nystr\"{o}m centres, which define a low-dimensional subspace of the RKHS, onto which the dataset is projected \cite{Williams2000,ahmad2023sketch}. The Nystr\"{o}m approximation is accurate under the assumptions that an appropriate sampling is carried out and the kernel matrix has a low rank, where the latter is often satisfied, e.g., for Gaussian kernels whose Gram matrix eigenvalue spectrum rapidly decays \cite{Williams2000b}. We remark that, concurrent to our work, \cite{Caldarelli2024} proposes sketched operator estimation, but only for the restrictive case of LTI RKHS dynamics that is recovered as a special case of our approach (cf. Table \ref{tab:cKORpower}).
 
To put our developments in perspective, recall that the popular parametric bEDMD \cite{Bruder2021b} or \cite{Guo2022} have the time-complexity of $\mathcal{O}(n_{\bm{\mathsf{z}}}^3(n_{{u}}+1)^3)$ -- where $n_{\bm{\mathsf{z}}}$ is the dictionary dimension -- due to a cubic magnification of complexity in case of control systems based on the dimensionality of the inputs. This is primarily due to a lack of a \textit{kernel trick} for the state-control couplings -- an inherent limitation of a parametric model. In contrast, the proposed combination of our nonparametric cKOR framework and random sketching that we will work out in this section preserves the ``kernel trick" at a reduced set of samples of cardinality $m$ to deliver a handy complexity reduction that is \textit{independent of input-dimensionality}, amounting to the time-complexity $\mathcal{O}({{m}}^3+{{m}}^2n)$, where ${{m}} \ll n$ is the number of inducing points -- mirroring the autonomous case \cite{Meanti2023}.

We will consider uniform random sampling because it is a simple algorithm that is generally applicable, and it is well-known that random projections are suitable for extracting a low-rank matrix approximation \cite{Tropp2017}. There exist more advanced sampling approaches, which have the potential to minimize the number of inducing points necessary, but reviewing these is out of scope for this work. Still, our sketched scalable estimators are not limited to uniform random sampling. We note that the concept of Nystr\"{o}m approximation is rigorously studied in the context of autonomous operators with KRR, PCR, and RRR estimators \cite{Meanti2023}. In the remainder of this section, we extend these results to \textit{control Koopman operator regression}. This approach starts by sampling a small subset of the data matrices with ${{m}}$ datapoints/columns, where ${{m}} \ll n$ is the number of inducing points. With these datasets, the {\em subsampling operators} $\nyES_X,\nyES_{+}: \spIN \rightarrow \Set{R}^{{{m}}}$ and $\nyES_Z: \spOUT \rightarrow \Set{R}^{{{m}}}$ are introduced to explicitly represent the orthogonal projections
\begin{subequations}\label{eq:projSketch}
    \begin{align}
   \projP_{+}\defeq& \adjoint{\nyES}_{+}(\adjoint{\nyES}_{+})^{\dagger}=\adjoint{\nyES}_{+}{\bm{K}}_{\widetilde{+}}^{\dagger}{\nyES}_{+}\defeq\widetilde{E}_+\widetilde{E}_+^*: \spIN\to {\spIN}\\
   \projP_Z\defeq& \adjoint{\nyES}_Z(\adjoint{\nyES}_Z)^{\dagger}=\adjoint{\nyES}_{Z}{\bm{K}}_{\widetilde{Z}}^{\dagger}{\nyES}_{Z}=\widetilde{E}_Z\widetilde{E}_Z^*: \spOUT \to{\spOUT},
\end{align} 
\end{subequations}
where ${\bm{K}}_{\widetilde{+}}\defeq \nyES_{{+}}\adjoint{\nyES}_{{+}}= [k_{X}(\bm{x}_+^{(i)},\bm{x}_+^{(j)})]_{i,j\in[m]}$, ${\bm{K}}_{\widetilde{Z}}\defeq \nyES_{{Z}}\adjoint{\nyES}_{{Z}}= [k_{Z}(\bm{z}^{(i)},\bm{z}^{(j)})]_{i,j\in[m]}$ and $\widetilde{E}_{({\cdot})}$ denote partial isometries\footnote{The input and output space partial isometries are $\widetilde{E}_+=\widetilde{S}_+^*(\widetilde{S}_+\widetilde{S}_+^*)^{\nicefrac{\dagger}{2}}\defeq \widetilde{S}_+^*{\bm{K}}_{\widetilde{+}}^{\nicefrac{\dagger}{2}},\widetilde{E}_Z=\widetilde{S}_{Z}^*(\widetilde{S}_{Z}\widetilde{S}_{Z}^*)^{\nicefrac{\dagger}{2}}\defeq \widetilde{S}_{Z}^*{\bm{K}}_{\widetilde{Z}}^{\nicefrac{\dagger}{2}}$, respectively.}, satisfying $\widetilde{E}_{({\cdot})}^*\widetilde{E}_{({\cdot})}=\bm{I}$. 
Similar to \cite{Meanti2023}, the variational problems \eqref{eq:infKRR}, \eqref{eq:KRR} are projected, using \eqref{eq:projSketch}, onto a lower-dimensional subspace, leading to \textit{sketched} risk minimization
    \begin{align}\label{eq:nyinfKRR}
 {\infEstim}_{m,{\reg}}
& {{\defeq}}\argmin_{\Estim \in \HS{\spIN,\spOUT}}{\!\!\!\!\!{\mathsfit{R}}(\projP_{Z}G\projP_{+}) {+} {\reg}\norm{\Estim}^{2}_{{\rm HS}}} \tag{${\txt{\textsc{s}r\textsc{rm}}}$}
\end{align} 
and its empirical risk counterpart
\begin{align}
\label{eq:nyKRR}
  \widehat{\infEstim}_{m,{\reg}}
& {{\defeq}}\argmin_{\Estim \in \HS{\spIN,\spOUT}}{\!\!\!\!\!\widehat{\mathsfit{R}}(\projP_{Z}G\projP_{+}) {+} {\reg}\norm{\Estim}^{2}_{{\rm HS}}}\notag\\
&= (\projP_Z\aES_Z \ES_Z\projP_Z+n{\reg} I)^{-1}\projP_Z\aES_Z \ES_+\projP_+.\tag{$\widehat{\txt{\textsc{s}r\textsc{rm}}}$}
\end{align}
\begin{algorithm}[t]
\caption{\textsc{Ny--cKOR Sketched Model Learning}}
\label{alg:nystrom-ckor}
\begin{algorithmic}
\algoSepIn
  \Statex {\bfseries Input:} Data $\{(\bm{x}^{(i)},\bm{u}^{(i)},\bm{x}^{(i)}_{+})\}_{i=1}^n$, subsample sets $\{\widetilde{\bm{x}}^{(i)}\}_{i=1}^m$, $\{\widetilde{\bm{u}}^{(i)}\}_{i=1}^m$, $\{\widetilde{\bm{x}}_+^{(i)}\}_{i=1}^m$, kernels $k_{X},k_{U}$, observable $\bm{y} \in \RKHS_X^{n_y}$, regularization $\gamma>0$.
  \algoSep
  \Statex{\texttt{/* Nonparametric Representation */}}
    \Statex \textbf{let:} ${\bm{I}}=\mathrm{diag}({\bm{1}})$, \; ${\bm{1}}\in\mathbb{R}^m$; $\widetilde{\bm{I}}=\mathrm{diag}(\widetilde{\bm{1}})$, \; $\widetilde{\bm{1}}\in\mathbb{R}^m$
    \Statex \textbf{let:} $\bm{k}_{\widetilde X}(\bm{x})_j=k_{X}(\bm{x},\widetilde{\bm{x}}^{(j)})$, \; 
                         $\bm{k}_{\widetilde U}(\bm{u})_j=k_{U}(\bm{u},\widetilde{\bm{u}}^{(j)})$
    \Statex \textbf{let:} $(\bm{K}_{\widetilde X})_{ij}=k_{X}(\widetilde{\bm{x}}^{(i)},\widetilde{\bm{x}}^{(j)})$,\;
                         $(\bm{K}_{X\widetilde X})_{ij}=k_{X}(\bm{x}^{(i)},\widetilde{\bm{x}}^{(j)})$
    \Statex \textbf{let:} $(\bm{K}_{\widetilde U})_{ij}=k_{U}(\widetilde{\bm{u}}^{(i)},\widetilde{\bm{u}}^{(j)})$, \;
                         $(\bm{K}_{U\widetilde U})_{ij}=k_{U}(\bm{u}^{(i)},\widetilde{\bm{u}}^{(j)})$
    \Statex \textbf{let:} $(\bm{K}_{+\widetilde{+}})_{ij}=k_{X}(\bm{x}_+^{(i)},\widetilde{\bm{x}}_+^{(j)})$, \;
                         $(\bm{K}_{\widetilde{+}})_{ij}=k_{X}(\widetilde{\bm{x}}_+^{(i)},\widetilde{\bm{x}}_+^{(j)})$
    \Statex \textbf{let:} $\widetilde{\bm{K}}_+\in\Set{R}^{m\times m}$, \;
                         $(\widetilde{\bm{K}}_+)_{ij}=k_{X}(\widetilde{\bm{x}}_+^{(i)},\widetilde{\bm{x}}^{(j)})$
    \Statex \textbf{let:} $\bm{K}_{\widetilde Z}=\bm{K}_{\widetilde X}\odot(\widetilde{\bm{1}}\widetilde{\bm{1}}^\top+\bm{K}_{\widetilde U})$
    \Statex \textbf{let:} $\bm{K}_{Z\widetilde Z}=\bm{K}_{X\widetilde X}\odot(\bm{1}\widetilde{\bm{1}}^\top+\bm{K}_{U\widetilde U})$
  \algoSep
  \Statex{\texttt{/* Model Learning */}}
    \Statex \textbf{compute:} $\bm{H}=\bm{K}_{Z\widetilde Z}^\top\bm{K}_{Z\widetilde Z}+n\gamma\,\bm{K}_{\widetilde Z}$
    \Statex \textbf{compute:} ${\widetilde{\bm{K}}_{\gamma}}^{-1}=\bm{H}^\dagger\bm{K}_{Z\widetilde Z}^\top\bm{K}_{+\widetilde{+}}\bm{K}_{\widetilde{+}}^\dagger$
    \Statex \textbf{compute:} $\bm{A}=({\widetilde{\bm{K}}_{\gamma}}^{-1}\widetilde{\bm{K}}_+)^\top$
    \Statex \textbf{let:} $\mathsfit{B}\bm{\mathsf v}=\mathrm{diag}(\bm{\mathsf v})\bm{A} \in \Set{R}^{m\times m}$
    \Statex \textbf{let:} $\bm{\mathsf z}(\bm{x},\bm{u})=\bm{k}_{\widetilde X}(\bm{x})\odot(\widetilde{\bm{1}}+\bm{k}_{\widetilde U}(\bm{u})) \in \Set{R}^{m}$
  \algoSep
  \Statex{\texttt{/* Inference */}}
    \Statex \textbf{let:} $\widetilde{\bm{Y}}_+=[\bm{y}(\widetilde{\bm{x}}_+^{(1)}),\dots,\bm{y}(\widetilde{\bm{x}}_+^{(m)})]^\top$
    \Statex \textbf{compute:} $\bm{C}=({\widetilde{\bm{K}}_{\gamma}}^{-1}\widetilde{\bm{Y}}_+)^\top \in \Set{R}^{m\times n_y}$
  \algoSep
  \Statex {\bfseries Return:} matrices $(\bm{A},\bm{C})$, map $\mathsfit{B}$, lifting $\bm{\mathsf z}(\cdot,\cdot)$.
\end{algorithmic}
\end{algorithm}
Same as before, we will arrive at an exact finite-dimensional reformulation after some algebra.
\begin{proposition}\label{prop:NYcKORkrrSHORT}
 Consider control $\RKHS_U$ and state $\RKHS_X$ RKHSs, forming the domain $\RKHS_X$ and range $\RKHS_Z \defeq \spIN \oplus (\RKHS_U \otimes \RKHS_X)$ for $G$ in \eqref{eq:MSEriskmin}. Given an observable of interest $\bm{y} \in \RKHS^{n_y}_X$, the repeated application of \eqref{eq:nyKRR} is equivalent to a finite-dimensional state-space model defined by Algorithm \ref{alg:nystrom-ckor} \& \ref{alg:ckorpred}.
\end{proposition}
\begin{proof}
Writing out the estimator \eqref{eq:nyKRR} gives
\begin{subequations}
    \begin{align}
& (\projP_Z\aES_Z \ES_Z\projP_Z+n{\reg} I)^{\!\!^{-1}}\projP_Z\aES_Z \ES_+\projP_+ \qquad \qquad \\
=& \widetilde{E}_Z({{\bm{K}}_{\widetilde{Z}}^{\nicefrac{\dagger}{2}}}\bm{K}^\top_{{Z}\widetilde{Z}}\bm{K}_{{Z}\widetilde{Z}}{{\bm{K}}_{\widetilde{Z}}^{\nicefrac{\dagger}{2}}}{+}n{\reg} I)^{\!\!^{-1}}\widetilde{E}_Z^*\aES_Z \ES_+\widetilde{E}_+\widetilde{E}_+^* \\
=& \widetilde{S}_{Z}^*{(\bm{K}^\top_{{Z}\widetilde{Z}}\bm{K}_{{Z}\widetilde{Z}}{+}n{\reg} {\bm{K}}_{\widetilde{Z}})^{\dagger} \bm{K}^\top_{{Z}\widetilde{Z}}\bm{K}_{{+}\widetilde{+}}{\bm{K}}_{\widetilde{+}}^{{\dagger}}}\widetilde{S}_{+}
\end{align}
\end{subequations}
where we used $P_Z=\widetilde{E}_Z\widetilde{E}_Z^*$ and the push-through identity together with $\widetilde{E}_Z^*\widetilde{E}_Z=\bm{I}$. By following the proof of \cite[Proposition C.2]{Meanti2023} we have ${{{\widetilde{\bm{K}}_{\gamma}}^{-1}}}{\defeq}  (\bm{K}^\top_{{Z}\widetilde{Z}}\bm{K}_{{Z}\widetilde{Z}}+n{\reg} {\bm{K}}_{\widetilde{Z}})^{\dagger} \bm{K}^\top_{{Z}\widetilde{Z}}\bm{K}_{{+}\widetilde{+}}{\bm{K}}_{\widetilde{+}}^{{\dagger}}$. Finally, we arrive at Algorithm \ref{alg:ckorpred} via Algorithm \ref{alg:nystrom-ckor} following the proof of Proposition \ref{prop:cKORkrrSHORT}.
\end{proof}
\section{Lifted State-Dimension Reduction}\label{sec:POD}
For cKOR and Ny-cKOR, the lifted state dimension scales with the dataset cardinality $n$ and inducing points ${m}$, respectively.
For relatively few inducing points, the lifted state dimension can still be high, posing challenges for efficient control design and real-time execution on low-level hardware.

A compelling approach for ordering and reducing the lifted states is based on proper orthogonal mode decomposition (POD), because it has been successfully applied in obtaining low-dimensional representations based on large-scale datasets in many applications \cite{Chatterjee2000}. The dynamic mode decomposition (DMD) algorithm actually makes use of this reduction \cite{schmid_2010} where it involves taking an SVD of the state ``data matrix" -- analogous to the SVD of the empirical embedding operator $\adjoint{\ES}_X$ in our RKHS setting, which ranks the orthogonal structures of this matrix based on the singular values. With this ranking, the $r$ dominant modes/coordinates can be selected to describe the dynamical behavior of the underlying system. 

The aforementioned reduction approach in the context of cKOR (and Ny-cKOR), starts by taking the $r$-truncated SVD of the kernel matrix $\SVDr{\bm{K}}= \bm{V}_r \bm{\Sigma}_r \bm{V}_r^{\intercal}$, with the POD modes $\bm{V}_r \in \Set{R}^{N_{\bm{z}} \times r}$ and $\bm{\Sigma}_r=\text{diag}(\sigma_1,\dots, \sigma_r)$.\footnote{We select $r$ thorough a threshold 
$\tau$ such that  $\nicefrac{\sum^{r}_{i=1}\sigma_i^{2}}{\sum^{n}_{j=1}\sigma_j^{2}} \cdot 100\% \leq \tau$, but there are various other methods for choosing $r$, see \cite{FALINI2022100064}.} With these POD modes, a bilinear lifted dynamics \eqref{eq:BilinMOD} can be transformed into the following reduced form ${{{\bm{\mathsf{z}}}}}_{k+1} = (\bm{V}_r^{\intercal}{\bm{{A}}}\bm{V}_r+ \lilsum^{n_u}_{i=1}\innerprod{\bm{e}_i}{\bm{u}_k}{\bm{V}_r^{\intercal}{\bm{{B}}}_i}\bm{V}_r){{{\bm{\mathsf{z}}}}}_k$.
In context of cKOR, this method is coined as reduced cKOR ($r$-cKOR). For Ny-cKOR, the reduction approach is identical to cKOR, with the difference that the truncated SVD is applied to the kernel matrix $\SVDr{\widetilde{\bm{K}}}$. In this case, the method is coined as reduced Ny-cKOR ($r$-Ny-cKOR).
\section{Numerical Examples}\label{sec:NumEx}
In this section, we present numerical studies to illustrate the implications of the theoretical results and showcase the advantages of our cKOR approach in practice. For comparison, the bilinear EDMD baseline \cite{Bruder2021,otto2023learning} is used with the same data- or subsample-based dictionary $\mathrm{span}\{k_{X}(\bm{x}^{(1)},\bm{x}),\dots,k_{X}(\bm{x}^{(n)},\bm{x})\}$. To match the baselines we use a linear control kernel $k_{U}(\bm{u},\bm{u}^\prime) = \innerprod{\bm{u}}{\bm{u}^\prime}$.
\preSection\subsection{Model learning for the Duffing oscillator}\label{subsec:2_num_examples}
As the first example, a controlled damped Duffing oscillator described by the state-space representation 
\begin{equation}
    \bm{\Dot{x}}=\begin{bmatrix}
    x_2 \\ x_1-x_1^{3}-0.5x_2
    \end{bmatrix}+\begin{bmatrix}
    0 \\ 2+\sin{(x_1)}
    \end{bmatrix}u,\label{eq:Control_Damp_Duff}
\end{equation}
is used, where the state is simulated using RK4 numerical integration with a step size $T_{\mathrm{s}} = 0.01s$ and $\bm{x}$ is sampled with the same $T_{\mathrm{s}}$. The control is applied in a synchronized zero-order-hold (ZOH) manner.
\subsubsection*{Prediction performance w.r.t. hyperparameters} First we compare the cKOR, Ny-cKOR and bEDMDc approaches over a range of the hyperparameters ${\mu}\in \Set{R}_{+}$ of the Gaussian (RBF) kernel $k_{X}(\bm{x},\bm{x}^\prime)=\operatorname{e}^{\nicefrac{-\|\bm{x}-\bm{x}^\prime\|_{2}^{2}}{{\mu}}}$ in terms of the 
$T_{\mathrm{test}}$-ahead prediction performance, quantified using the root mean square error (RMSE) $({\nicefrac{1}{T_{\mathrm{test}}}\sum_{t=1}^{T_{\mathrm{test}}} \lVert \bm{y}_{t}-\bm{\hat{y}}_{t}\rVert^2_2})^{\nicefrac{1}{2}}$, where $T_{\mathrm{test}}$ denotes the number of steps and $\bm{y}_{t}-\bm{\hat{y}}_{t}$ is the difference between the true system response and the predicted solution. For the training dataset, $200$ trajectories with $n=1000$ samples ($10$ sec) are generated, starting from a $14$ by $14$ grid of initial conditions within the limits $|x_1|\leq 2.25$ and $|x_2|\leq 2.25$. For the test dataset, 40 trajectories of $T_{\mathrm{test}}=100$ samples ($1s$) are generated, starting from random initial conditions sampled within the limits $|x_1|\leq 2$ and $|x_2|\leq 2$ using a uniform distribution. Both datasets are generated using uniform random input sequences within the interval $[-2,2]$. From the training dataset, ${{m}}=200$ inducing points are randomly sampled. The considered approaches are used with a regularization parameter $\gamma=10^{-9}$. 
Figure \ref{fig:RMSE_vs_mu_cKOR_bEDMDc_full} confirms the superior accuracy of our nonparametric cKOR estimator as it reaches a significantly lower error than bEDMDc across ${\mu}$ values -- showing a greater hyperparameter range of increased accuracy. Additionally, when lowering the regularization, the accuracy of our cKOR estimate increases, achieving up to an \textit{order of magnitude better accuracy than bEDMDc}, cf. Appendix \ref{app:appendix_exp_results}. 
Even when we reduce complexity by projecting on a subset of data for our sketched Ny-cKOR estimator, a similar, but slightly reduced advantage can be observed in Figure \ref{fig:RMSE_vs_mu_Ny_cKOR_bEDMDc_tilde}. 
\begin{figure}[t]
\centering
  \centering
  \includegraphics[width=\plotfactor\columnwidth]{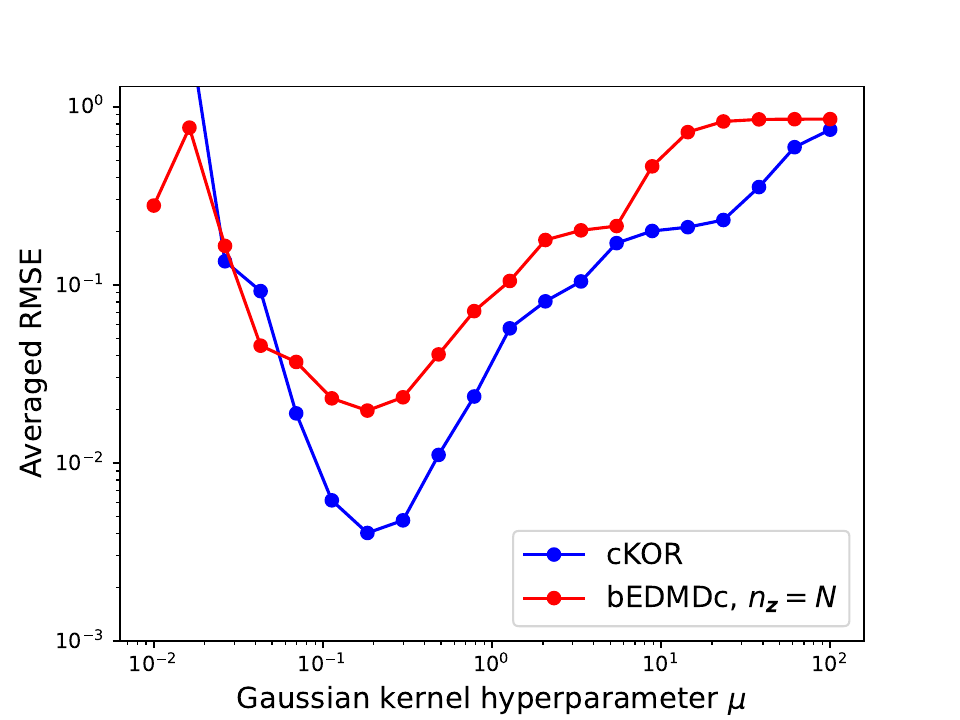}
  \caption{Averaged RMSE of the 1-step-ahead prediction for the cKOR, Ny-cKOR, and bEDMDc models over the test set for various choices of the kernel width ${\mu}$ for the case of predictor dimension $n_{\mathsf{\bm{z}}}=1000$.}
  \label{fig:RMSE_vs_mu_cKOR_bEDMDc_full}
    \postFloat
\end{figure}%
\begin{figure}
  \centering
  \includegraphics[width=\plotfactor\columnwidth]{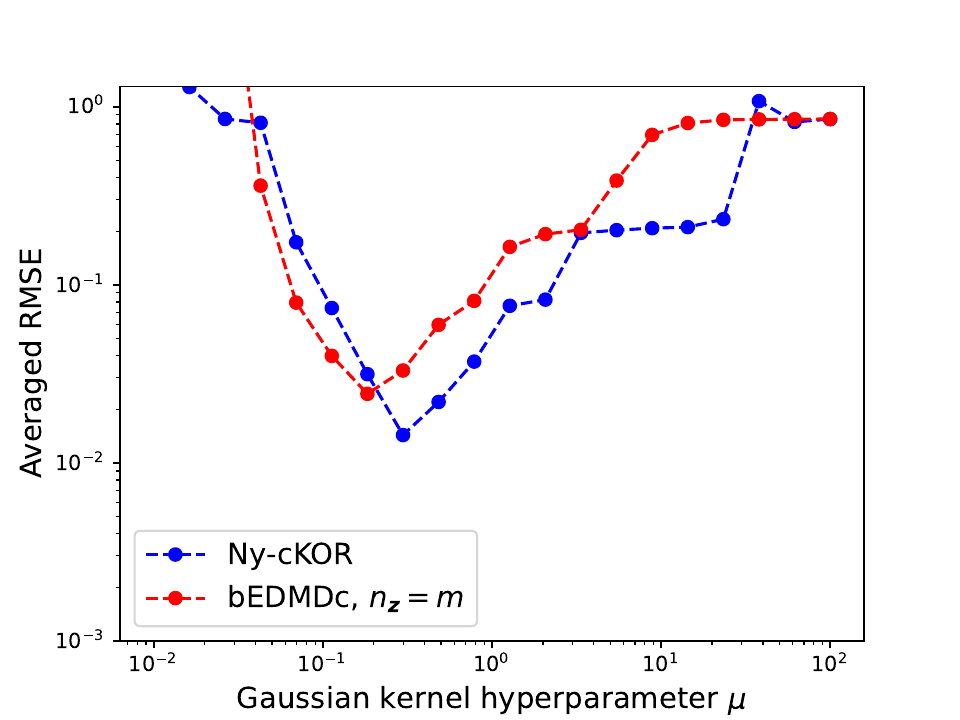}
  \caption{Averaged RMSE of the 1-step-ahead prediction for the cKOR, Ny-cKOR, and bEDMDc models over the test set for various choices of the kernel width ${\mu}$ for the case of predictor dimension $n_{\mathsf{\bm{z}}}=200$.}
  \label{fig:RMSE_vs_mu_Ny_cKOR_bEDMDc_tilde} 
\label{fig:RMSE_vs_mu_cKOR_bEDMDc_full_and_RMSE_vs_mu_Ny_cKOR_bEDMDc_tilde}
  \postFloat
\end{figure}
\looseness=-1
\subsubsection*{Statistical performance \& time-complexity evaluation} Next, the training data is varied in terms of the number of samples. To generate the data snapshots, the initial conditions are sampled from a square and equidistant grid within the limits $|x_1|, |x_2|\leq 2$ and the system is driven for $2.0s$ by randomly generated control inputs following a uniform distribution within the limits $|u|\leq 2$. The test data consists of 20 trajectories of length $2.0$ sec ($T_{\mathrm{test}}=200$) with the same initial condition generation, but driven by an input sequence of $u_t=2\sin{(10 \pi t)}$. 
\begin{figure}[t]
\centering
  \centering
  \includegraphics[width=\plotfactor\columnwidth]{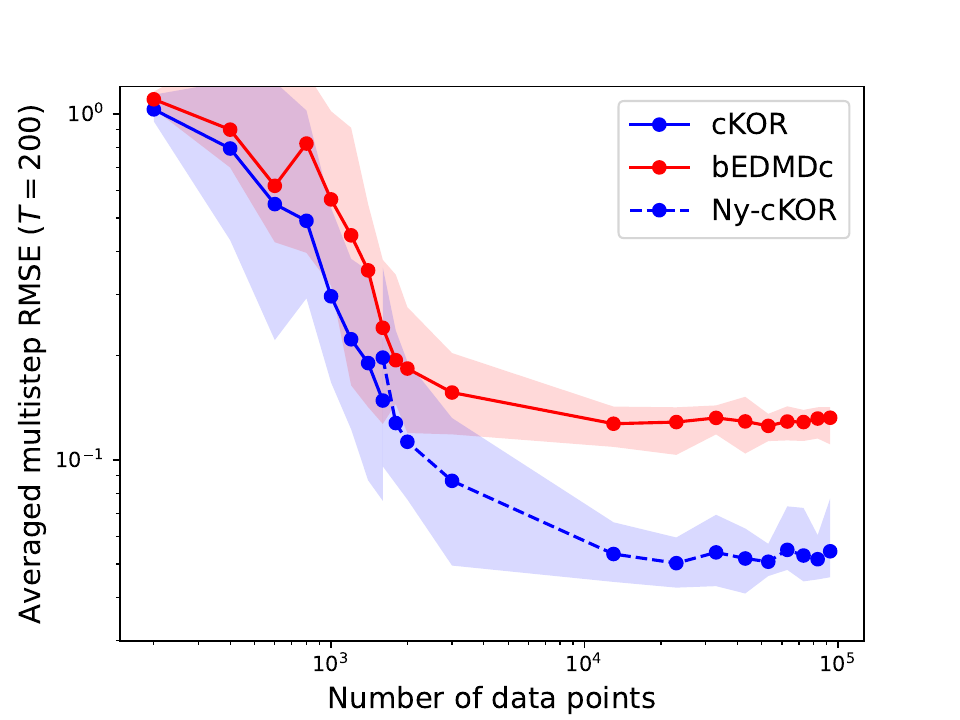}
  \caption{RMSE of the 200-step-ahead prediction for the cKOR, Ny-cKOR and bEDMDc models w.r.t. increasing training data size. Our (Ny)-cKOR estimators attain significantly lower errors.}
  \label{fig:RMSE_vs_N_D_bKDMDc_bEDMDc_KRR_Ny}
  \postFloat
\end{figure}%
\begin{figure}
  \centering
  \includegraphics[width=\plotfactor\linewidth]{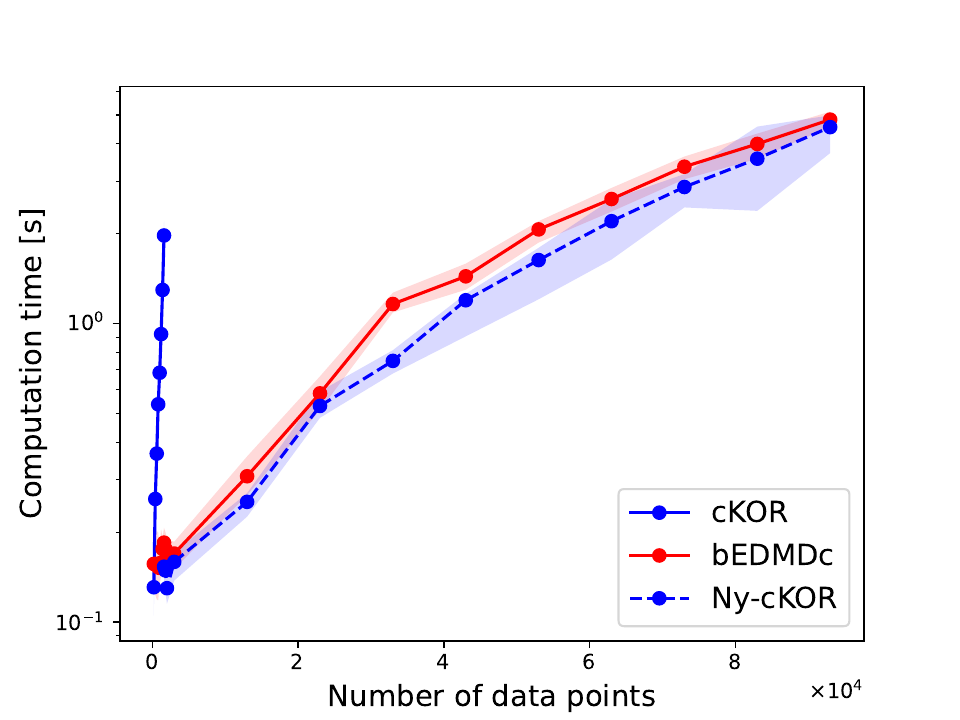}
  \caption{Computation time for the cKOR,
Ny-cKOR and bEDMDc models over increasing training data size. Our Ny-cKOR estimator effectively attains the complexity of the parametric approach.}
  \label{fig:Time_vs_N_D_bKDMDc_bEDMDc_KRR_Ny}
\label{fig:test1}
  \postFloat
\end{figure}
\looseness=-1
All the approaches use the Gaussian kernel to construct the lifted states with hyperparameter ${\mu}=0.25$ and a regularization parameter of $10^{-7}$. These values are empirically determined as ``optimal'' for the prediction RMSE, and choosing the same settings allows for a fair comparison between cKOR, Ny-cKOR, and bEDMDc. Note that the bEDMDc approach takes the inducing points as centers, which are 200 uniformly randomly sampled points from the training dataset. 
Figure \ref{fig:RMSE_vs_N_D_bKDMDc_bEDMDc_KRR_Ny} illustrates the RMSE of the $T_{\mathrm{test}}$-step-ahead prediction averaged over the test trajectories versus the number of training datapoints. The solid lines represent the average RMSE over 20 runs and the shaded area gives the variation of the RMSE per run. For each run, a new training data set and inducing points are generated to provide statistically relevant results. Figure \ref{fig:Time_vs_N_D_bKDMDc_bEDMDc_KRR_Ny} shows the computation times, i.e., the estimation time of the predictor along with the $n$-step-ahead prediction/rollout time. Strikingly, both the average RMSE of cKOR and Ny-cKOR stay below that of bEDMDc, confirming the inherent advantages of estimators derived using a nonparametric paradigm. This also illustrates the common bottleneck of full KRR estimators well over the number of datapoints, since the full cKOR scales with $\mathcal{O}(n^3)$, compared to bEDMDc with $\mathcal{O}((n_{\mathsf{\bm{z}}}(n_{{u}}+1))^3)$ and Ny-cKOR $\mathcal{O}({{m}}^3)$. In line with our expectation, Ny-cKOR continues the trend of cKOR for larger datasets, as the computation time becomes intractable for the full cKOR estimator. 
 For a single input, the bEDMDc complexity is comparable to Ny-cKOR, it is important to stress that Ny-cKOR is substantially more computationally efficient than bEDMDc for higher input dimensions -- by a factor of $(n_{{u}}+1)^3$ -- as it does not require taking a tensor product of features and inputs.
\preSection\subsection{Learning the high-dimensional Kalman vortex street}
\begin{figure*}[t!]
\centering
\includegraphics[width=0.7%0.66
\linewidth]{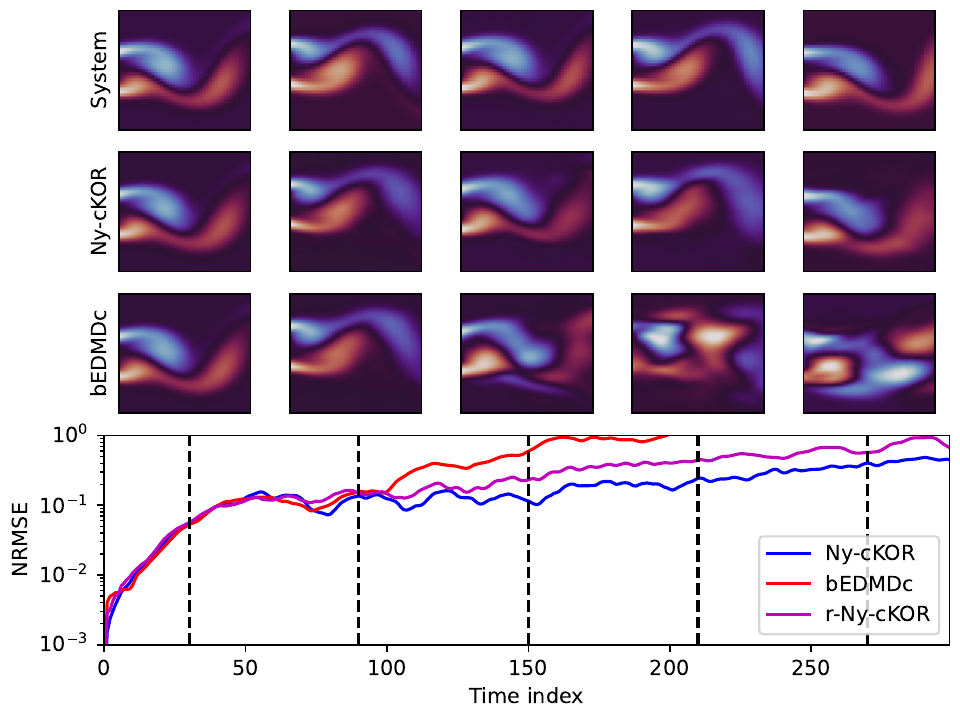}
  \caption{A prediction of the flow in the Kalman Vortex Street example, showing the reduced prediction error of our higher fidelity Ny-cKOR models.}
  \label{fig:Fluid_RMSE_over_time_vis_test0_225} 
  \postFloat
\end{figure*}
Tackling high-dimensional systems in a parametric manner often comes with challenges, as the suitable basis for representation is of critical importance.
Through this simulation study, we want to showcase the superiority of our nonparametric learning paradigm even when it is based on a fraction of the data points, which is highly important for scalability.  
 \begin{table}[t!]
  \centering
  \normalsize
      \caption{Test normalized RMSE (NRMSE) for the estimated models on the actuated Karman vortex street for multi-step prediction over $T_{\text{test}}=3T_{\text{valid}}$ from $20$ test-validation-train splits.
    }
    \begin{tabular}{r|ccc}
        \toprule
        & \textbf{{Ny-cKOR}}   & \textbf{$r$-Ny-cKOR}  &  bEDMDc   \\ 
        \midrule
         \!\!\textsc{nrmse} & $\bm{0.157}\bm{{\pm} 0.049}$  &  $\bm{0.206{\pm} 0.109}$  &  $0.412{\pm} 0.976$\\
        \bottomrule
    \end{tabular}
\label{tab:PredErr}
    \postFloat
\end{table}
 Figure \ref{fig:Flow_Oscillating_Cylinder} schematically illustrates the considered nonlinear system, which generates a flow as a result of transverse non-slip movement of an oscillating cylinder as input. This flow exhibits vortex shedding, causing vortex-induced vibrations on the structure, which accelerate material fatigue and may lead to failure \cite{Bearman1984}. In \cite{DECUYPER2020}, the considered system is created and simulated using the Computational Fluid Dynamics (CFD) environment OpenFOAM. For the data generation, we refer to \cite{DECUYPER2020}. The data-driven model learning is performed with the same setting as in \cite{BeintemaPhD2024}. 
 \begin{wrapfigure}{l}{\fpeval{\plotfactor*0.5}\linewidth}
\centering
\includegraphics[width=0.99\linewidth]{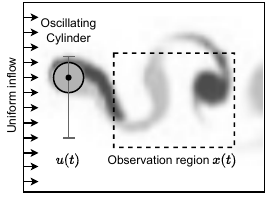}
\postFloat
  \caption{An oscillating cylinder in a uniform flow.}
  \label{fig:Flow_Oscillating_Cylinder}
\end{wrapfigure} For completeness, these settings are repeated here: the pressure, horizontal velocity and vertical velocity are observed in the rectangular wake region of $41 \times 45$, which amounts to a high state dimension of $n_{{x}}=5535$. The flow conditions correspond to a Reynolds number of $100$ and a Strouhal number of $0.167$. The dataset consists of 11 timeseries of length $T = 1520$, with samples measured at $50\text{Hz}$. These 11 timeseries correspond to 11 input sequences with a swept-sine input profile with different amplitudes and cylinder diameter ratios. The timeseries are randomly split into: 6 series for training, 3 for validation, and 2 for testing. For a statistically significant comparison, we assess the prediction performance of the methods on the test data over 20 randomly assigned splits. The training, validation, and test datasets are normalized to constrain the states and inputs to values $\leq 1$.
 For a fair comparison, all approaches use the same 400 inducing points for learning, meaning bEDMDc uses 400 RBF centres based on the inducing points of Ny-cKOR (and $r$-Ny-cKOR). Also all models are fitted using a hyperparameter and regularizer grid search on the validation data with the grids ${\mu} \in \{0.1, 0.5, 1, 10, 20, \ldots 60, 150, 175, \ldots 400\}$ and $\gamma \in \{10^{-11}, 10^{-10}, \ldots, 10^{-6}\}$ for the multi-step ($T_{\text{valid}}=100$) state prediction RMSE. The rank $r$ of the POD reduction is obtained for $\tau = 99.99\%$ (cf. Section \ref{sec:POD}). 

As shown in Table \ref{tab:PredErr}, on average, both models from $r$-Ny-cKOR and Ny-cKOR significantly outperform those of bEDMDc. Strikingly, \textit{the error variance} compared to bEDMDc \textit{for our Ny-cKOR and r-Ny-cKOR models is $20{\times}$ and $10{\times}$ smaller}, respectively. Figure \ref{fig:Fluid_RMSE_over_time_vis_test0_225} shows examples of significant performance loss of bEDMDc models due to a large error variance. By just comparing the flow plots, it becomes clear that bEDMDc quickly deteriorates and does not resemble any aspect of the flow, as opposed to Ny-cKOR, which stays quite accurate over the entire horizon. The reduced model of $r$-Ny-cKOR comes with an offset to Ny-cKOR, but does not exhibit the performance loss of bEDMDc. 
This example demonstrates the superiority of the nonparametric paradigm through the significantly better prediction accuracy of our ($r$-)Ny-cKOR models for an unknown high-dimensional control system. 
\preSection\subsection{Model predictive control with cKOR predictors}
Here we integrate our control operator predictors from (Proposition \ref{prop:cKORkrrSHORT} and \ref{prop:NYcKORkrrSHORT}) in an \textit{iterated} LPV-MPC scheme \cite{hoekstra2023computationally}, that we call \textit{control Koopman operator LPV-MPC} ({c\textsc{k}o\textsc{lpv-mpc}}), whose description is delegated to Appendix \ref{sec:co_LPV_MPC}. In a nutshell, we solve a single QP at every time-instant where the {c\textsc{k}o\textsc{lpv-mpc}} updates the scheduling iteratively over the simulation time in a receding horizon manner. This may cause some loss of performance, but convergence is still observed in practice, similar to SQP schemes. 
\subsubsection*{Damped Duffing oscillator} 
First, we consider the MPC design for the Duffing oscillator (\ref{eq:Control_Damp_Duff}), whose autonomous dynamics exhibit two stable equilibrium points while the origin is unstable. Figure~\ref{fig:Duffing_State_plot_full} shows the vector field of the system, illustrating whether the control trajectories exploit the dynamics to achieve performance. The simulated scenario starts from the initial condition $[1\ 1]^{\intercal}$, with a state reference switching between $[-1\ 0]^{\intercal}$ for $9s$ and $[0\ 0]^{\intercal}$ for $3s$. The weighting matrices are chosen as $\bm{Q}=\text{diag}(6,1)$, $R=5$, and $\bm{Q}_T=100\bm{Q}$, with a horizon of $T=100$ steps and sampling time $T_{\mathrm{s}}=0.01s$. The constraints are $-2 \leq u \leq 2$, $-3 \leq x_1 \leq 3$, and $-3 \leq x_2 \leq 3$, which reward model accuracy while tracking the set points.

As MPC baselines, we include a linear MPC ({\textsc{lmpc}}) obtained by linearizing the system around the origin, a sequential linearization-based predictive control approach\footnote{It relies on linearizations along predicted states and inputs, corresponding to a non-robustified LTV-MPC of \cite{berberich2022linear}.} ({\textsc{lpv-mpc}}), and a nonlinear MPC ({\textsc{nmpc}}) with access to the exact model of \eqref{eq:Control_Damp_Duff}, which serves as the ground truth. Our proposed {c\textsc{k}o\textsc{lpv-mpc}} does not use the system equations but instead relies on training data ($n=704$) collected in the state domain $-2 \leq x_1 \leq 2$, $-2 \leq x_2 \leq 2$, and input domain $-2 \leq u \leq 2$. The bilinear model is constructed via $r$-Ny-cKOR with 100 uniformly randomly sampled inducing points, $r=29$, and a hyperparameter grid search over validation data. For the MPC implementation, we rely on the state-of-the-art ForcesPRO\footnote{Freely-available academic licensing at \href{https://www.embotech.com/forcespro}{www.embotech.com/forcespro}.} solver \cite{FORCESPro,FORCESNLP} throughout. 

From Figures \ref{fig:Duffing_State_plot_full}, it can be observed that {c\textsc{k}o\textsc{lpv-mpc}}, {\textsc{lpv-mpc}}, and {\textsc{nmpc}} are almost identical, while they showcase a clear performance improvement compared to {\textsc{lmpc}}. Specifically, {\textsc{lmpc}} requires extra input effort between $1s$ and $3s$, because around $x_1=-0.2$ and $x_2=-1.25$, the {\textsc{lmpc}} solution goes against the vector field. The other approaches use the vector field to reach the setpoint and thus require less input. In addition, there is a large offset between the settled state of {\textsc{lmpc}} and the setpoint, which then leads to an additional input effort to stabilize the origin. The other approaches are almost identical, which implies that: 1) the system and/or control task is not challenging enough, and that 2) the bilinear Koopman model accurately identifies the nonlinear system.
\begin{figure}[t]
  \centering
  \includegraphics[width=\plotfactor\linewidth]{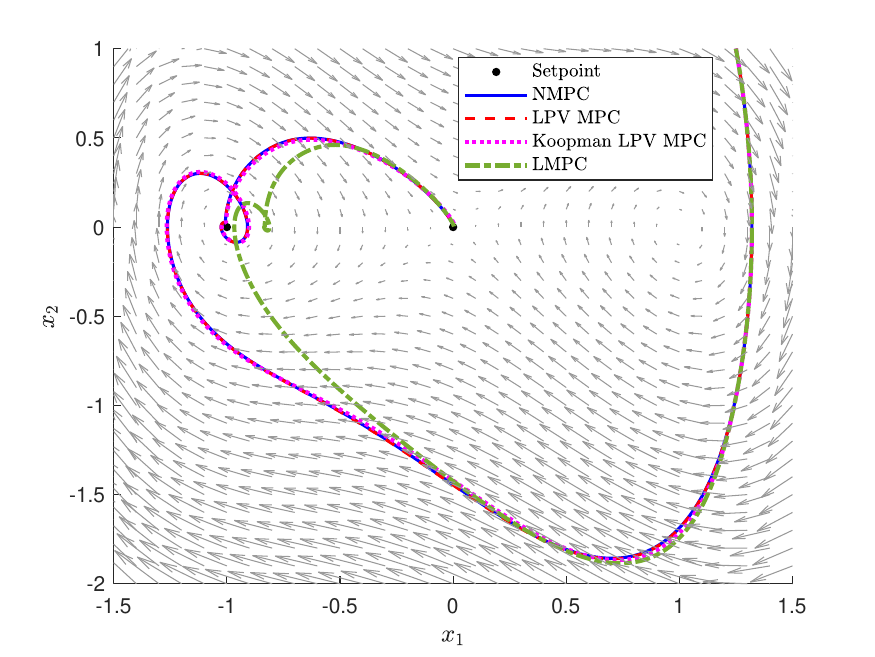}
  \caption{State-space response of the simulated Duffing oscillator under piecewise constant reference tracking for the MPC.}
  \label{fig:Duffing_State_plot_full}
  \postFloat
\end{figure}%
\looseness=-1
\subsubsection*{Unstable system with linearly uncontrollable origin} 
For the next example, we use the following Van der Pol oscillator
\begin{align}
\bm{\Dot{x}}=\begin{bmatrix}
    x_2 \\ -x_1-\frac{1}{2}x_2(1-x_1^2)
    \end{bmatrix}+\begin{bmatrix}
    0 \\ x_1 u
    \end{bmatrix},
\end{align}
which is simulated using RK4 numerical integration. Here, the state is measured with sampling time $T_{\mathrm{s}} = 0.05s$, and the input is applied in a synchronized ZOH manner.
This is an interesting example for three reasons: the origin is linearly uncontrollable, the optimal solution to drive the system from an arbitrary initial condition to the origin is known, i.e., the infinite-horizon optimal controller for the cost function $J(\bm{x},u)=x_2^2 + u^2$ is $u=-x_1 x_2$ \cite{Nevistic1996ConstrainedNO}. Thus, the MPC control task is to minimize the aforementioned cost function over a finite-horizon. In other words, this system can showcase the potential of nonlinear control techniques as opposed to controllers based on linear state-space models. We evaluate the performance on, the initial conditions $[-2,-2]^{\intercal}$, $[-2,2]^{\intercal}$, $[2,-2]^{\intercal}$ and $[2,2]^{\intercal}$ with $\bm{Q}=\text{diag}(0,1)$, $\bm{Q}_T=\bm{Q}$ and $R=1$ to match the above cost function. The horizon is chosen as the minimal one such that {\textsc{nmpc}} stabilizes the origin. This resulted in a horizon of 100 steps, which is relatively big and thus another indicator of a difficult-to-control system. The latter, in combination with being open-loop unstable in the considered region, complicates the data-gathering step for learning. For the training and the validation data, the {\textsc{nmpc}} controller is used to control the system to the origin with an exploratory uniform random disturbance within the interval $[-2,2]$. The hyperparameter and regularization parameter are obtained by employing a grid search on the validation data. The same initialization of the scheduling is used as in the previous example.
Figure \ref{fig:Van_der_Pol_state_plot_10s} shows the resulting control trajectories for {\textsc{nmpc}} and {c\textsc{k}o\textsc{lpv-mpc}}. For these settings, the {\textsc{lpv-mpc}} and {\textsc{lmpc}} controllers \textit{fail to stabilize the origin, due to limitations of linearization and the linearly uncontrollable property}. The latter clearly highlights an advantage for the {c\textsc{k}o\textsc{lpv-mpc}} scheme. However, the control trajectories of {c\textsc{k}o\textsc{lpv-mpc}} deviate from the trajectories of the {\textsc{nmpc}} with full exact model knowledge. The aforementioned deviations are quantified with RMSE: RMSE of {c\textsc{k}o\textsc{lpv-mpc}} is $2.83\cdot 10^{-1}$ and of {\textsc{nmpc}}: $1.01\cdot 10^{-1}$. Due to its data-driven nature, the model is inherently approximative, while the {\textsc{nmpc}} works with perfect system knowledge. The latter is illustrated in Figure \ref{fig:Van_der_Pol_time_plot_{U}_{X}0_2_min2} by requiring a higher input as opposed to following the vector field. Note that {c\textsc{k}o\textsc{lpv-mpc}} \textit{solves a single QP at every timestep and does not require any initial guess for the scheduling or employment of model-based planners}.
\begin{figure}[t]
  \centering
  \includegraphics[width=\plotfactor\columnwidth]{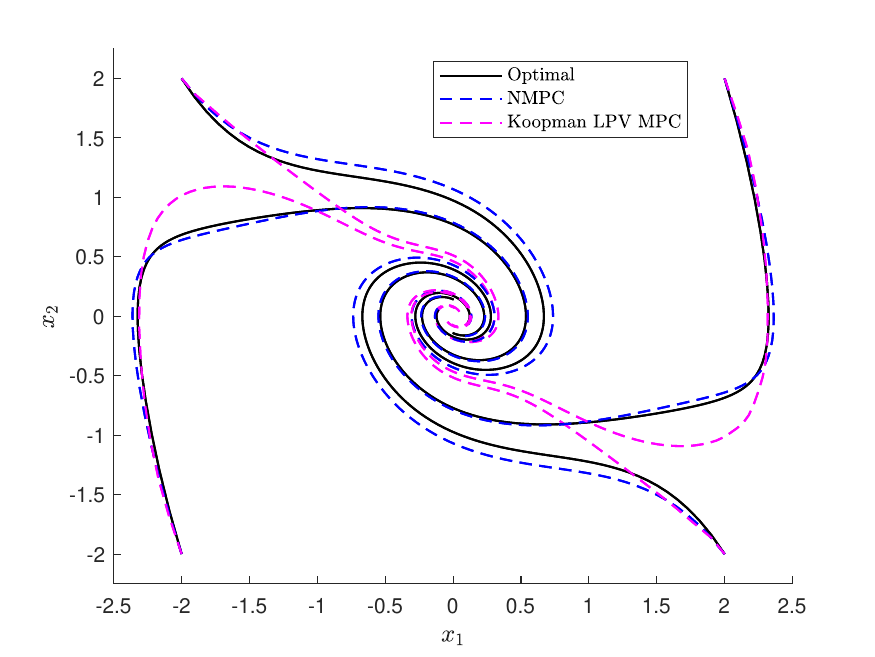}
  \caption{State-space responses of the simulated Van der Pol oscillator for the c\textsc{k}o\textsc{lpv-mpc} vs \textsc{nmpc}.}
  \label{fig:Van_der_Pol_state_plot_10s}
  \postFloat
\end{figure}%
\begin{figure}
  \centering
  \includegraphics[width=\plotfactor\linewidth]{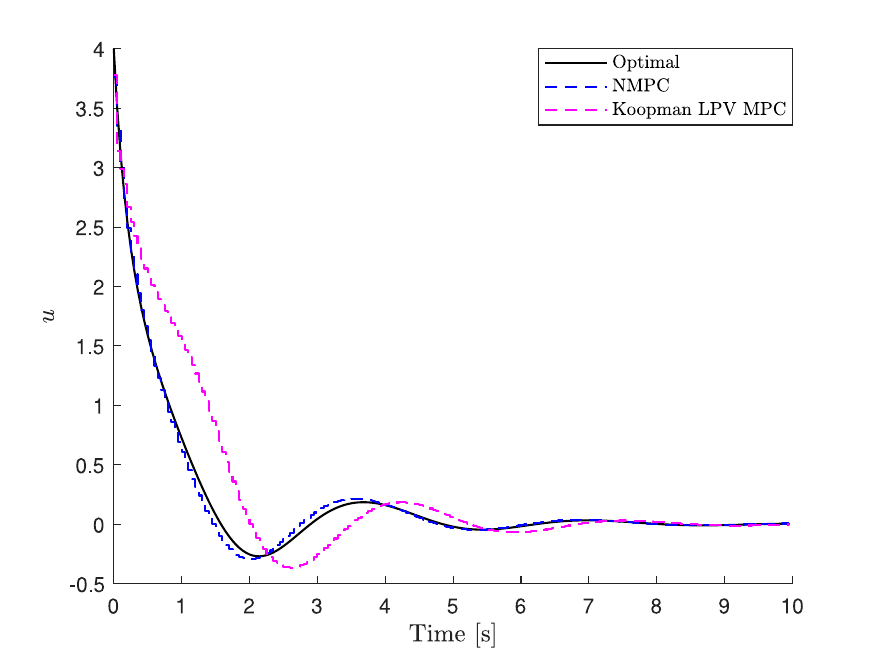}
  \caption{Comparing c\textsc{k}o\textsc{lpv-mpc} and \textsc{nmpc} control input signals form the initial condition $\bm{x}_0=[2,-2]^{\intercal}$ for the simulated Van der Pol.}
  \label{fig:Van_der_Pol_time_plot_{U}_{X}0_2_min2}
\label{fig:Van_der_Pol_state_plot_10s_and_Van_der_Pol_time_plot_{U}_{X}0_2_min2)}
\postFloat
\end{figure}
\looseness=-1
\section{Conclusion}\label{sec:Concl}
We introduce a novel framework for learning Koopman operators for control-affine systems in reproducing kernel Hilbert spaces (RKHS), grounded in risk minimization and infinite-dimensional regression.
By establishing the equivalence of various operator formulations, we enable the simultaneous use of vector-valued regression for learning and LPV Koopman forms for prediction and control. This equivalence demonstrates that control operators can be fully described using scalar-valued kernels, bridging a critical gap in existing operator representations and fully leveraging the available RKHS structure.
Our proposed empirical estimators are finite-rank operators over RKHSs that reduce exactly to finite-dimensional predictors, regardless of feature and input dimensions. Furthermore, we prove that our approach allows arbitrarily accurate operator norm approximations under minimal assumptions using finite-rank operators.
Additionally, we propose sketched estimators to improve the scalability of our method by reducing computational complexity to $\mathcal{O}(m^2)$ for large-scale problems, where $m$ may be much smaller than the data cardinality.
As implied by our theoretical analysis, the numerical experiments
demonstrate superior prediction accuracy compared to
bilinear EDMD, especially in high dimensions.
Finally, our learned models integrate seamlessly with LPV techniques for model predictive control, offering a viable alternative to typical MPC approaches.
\preSection
\section*{Acknowledgments}
The authors thank Nicolas Hoischen and Max Beier for their feedback on the manuscript and Jan Decuyper for sharing the details of the Karman vortex example.
\bibliographystyle{IEEEtran}
\bibliography{references_harmonized}
\preApdx
\appendix
\begin{figure}
\centering
\includegraphics[width=0.8\columnwidth]{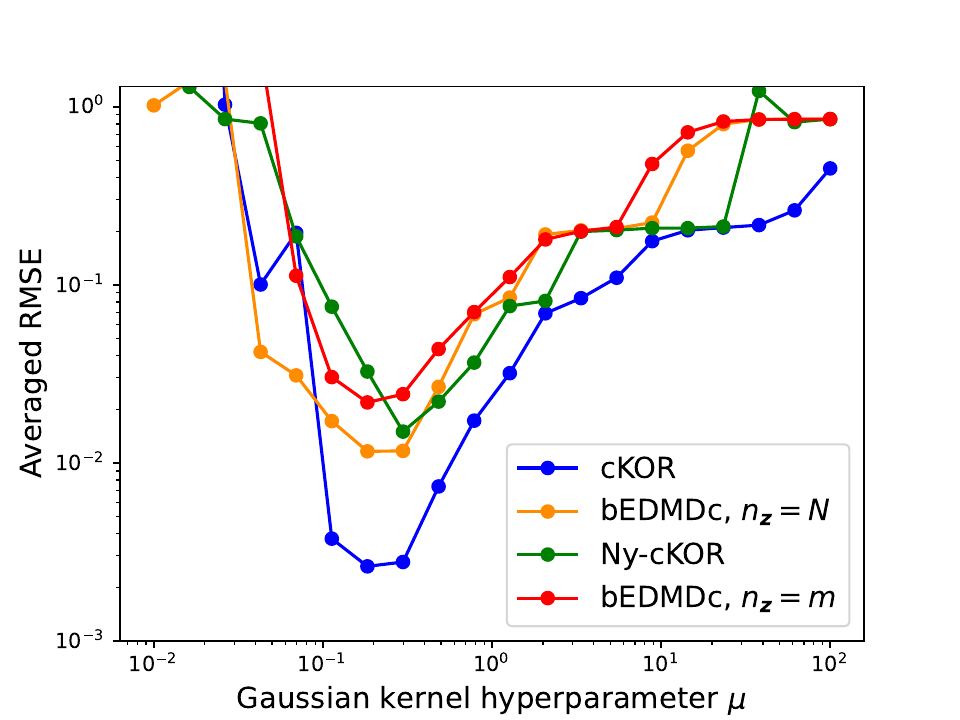}
  \caption{RMSE of 1-step-prediction for cKOR, Ny-cKOR and bEDMDc models over values of ${\mu}$ and regularizer $\gamma = 10^{-10}$ for the Duffing oscillator example \ref{subsec:2_num_examples}.}
\label{fig:RMSE_vs_mu_cKOR_bEDMDc_full_1e_06_and_RMSE_vs_mu_Ny_cKOR_bEDMDc_tilde_1e_06}
\postFloat
\end{figure}%
\subsubsection*{Additional ablation}\label{app:appendix_exp_results}
We expand the numerical study of the Duffing oscillator \ref{subsec:2_num_examples}, with an ablation study for an even lower level of regularization $\gamma=10^{-10}$. As shown in Figure \ref{fig:RMSE_vs_mu_cKOR_bEDMDc_full_1e_06_and_RMSE_vs_mu_Ny_cKOR_bEDMDc_tilde_1e_06}, our cKOR approach continues to significantly outperform the parametric approaches with its sketched Ny-cKOR version.
\preSection
\subsubsection*{Koopman-Based \textit{Iterated} LPV-MPC}\label{sec:co_LPV_MPC}
To provide an efficient scheme for controlling the original nonlinear system via the cKOR method, which provided surrogate models, we propose a model predictive control (MPC) approach that extends the \textit{iterated} LPV scheme of \cite{hoekstra2023computationally} to our control operator setting. 
 The control problem that we want to address is that given a measurement of the state $\bm{x}({k})$ of the original nonlinear system \eqref{eq:ncs} at time-instant $t$, based on a cKOR model, solve a predictive control problem on a finite time horizon $T$ with computation cost close to an LTI-MPC to obtain a control sequence $\{\bm{u}_{i|t}\}_{i=0}^{T-1}$ such that the predicted response of \eqref{eq:ncs} follows a prescribed reference trajectory. 
 Then, $\bm{u}_{0|t}$ is applied to the system and at the next time-instant ($t{+}1$), $\bm{x}(t{+}1)$ is measured to start the next control cycle. 
For our cKOR model, the {c\textsc{k}o\textsc{lpv-mpc}} optimization problem is
\begin{subequations}
    \begin{align}
\!\!\!\min_{\bm{u}_{0|t} {\cdots} \bm{u}_{T-1|t}}&\lilsum_{i=1}^{T}{ \left(\lVert{{{\bm{\mathsf{z}}}}}_{i|t}{-}\bm{\overline{\mathsf{z}}}_{i|t}\rVert^2_{\bm{Q}_{\mathsf{z}}} {+} \lVert\bm{u}_{i|t}{-}\bm{\overline{u}}_{i|t}\rVert^2_{\bm{R}} \right)}{+}q_{T|t}\label{eq:opt_cost}\\
\textrm{s.t.} \quad &{{{\bm{\mathsf{z}}}}}_{i+1|t}={\bm{A}}{{{\bm{\mathsf{z}}}}}_{i|t}+\mpcB(\scheduleVar_{i|t})\bm{u}_{i|t}, \label{eq:opt_eq_state}  \\
  \bm{x}_\mathrm{min}\leq & {\bm{C}}{{{\bm{\mathsf{z}}}}}_{j|t} \leq \bm{x}_\mathrm{max},
\quad\bm{u}_\mathrm{min}\leq \bm{u}_{i|t}\leq \bm{u}_\mathrm{max} \label{eq:opt_ineq}
\end{align}
\end{subequations}
where $\bm{Q}_{\mathsf{z}}={\bm{C}}^\top\bm{Q}{\bm{C}}$ and $q_{T|t}{=}\lVert{{{\bm{\mathsf{z}}}}}_{T|t}-\bm{\overline{\mathsf{z}}}_{T|t}\rVert^2_{\bm{Q}_{\mathsf{z}}(T)}$ is the terminal cost. The measured state at time $t$, i.e., $\bm{x}_{0|t}=\bm{x}(t)$, is lifted to determine ${{{\bm{\mathsf{z}}}}}_{0|t}=\bm{z}(\bm{x}_{0|t})$, while $\bm{\overline{x}}_{i|t}$ and $\bm{\overline{u}}_{i|t}$ denote the reference state and input. The matrices $\bm{Q}, \bm{Q}(T) \in \Set{R}^{n_{{x}} \times n_{{x}}}$ and $\bm{R} \in \Set{R}^{n_{{u}} \times n_{{u}}}$ are stage and terminal weights, tuned to user-specified tracking performance, and $\bm{x}_\mathrm{min}, \bm{x}_\mathrm{max}, \bm{u}_\mathrm{min}, \bm{u}_\mathrm{max}$ bound the state and input sequences.
To efficiently handle the bi-linearity of the cKOR model, a scheduling variable $\scheduleVar_{i|t}$ is introduced so that $\mpcB(\scheduleVar_{i|t}) = [{\bm{B}}_1\scheduleVar_{i|t} \mid \cdots \mid {\bm{B}}_{n_{{u}}}\scheduleVar_{i|t}]$.

The core idea is that at any given time-instant $t$, for a fixed scheduling sequence $\{\scheduleVar_{i|t}\}_{i=0}^{T-1}$, \eqref{eq:opt_eq_state} is used to formulate a linear MPC problem that can be solved efficiently as a quadratic program (QP). Then, the resulting control sequence $\{\bm{u}_{i|t}\}_{i=0}^{T-1}$ is used to forward simulate the cKOR model to compute a new sequence $\scheduleVar_{i|t}= \bm{\mathsf{z}}(i+t)$ on which a new sequence of control matrices $\tbm{B}(\scheduleVar_{i|t})$ is computed in a receding horizon fashion akin to the iterated MPC scheme of \cite{hoekstra2023computationally}.
\preBio
\begin{IEEEbiography}[{\bioPic{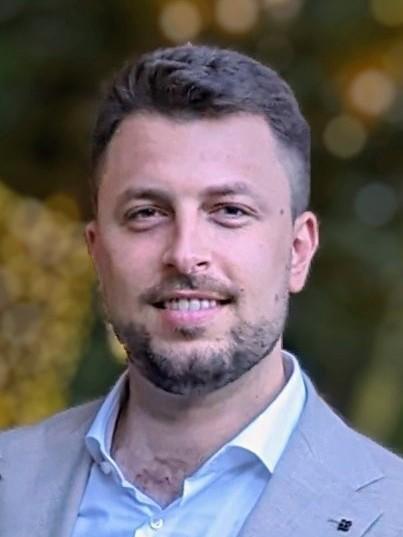}}]{Petar Bevanda}
	received his B.Sc. and M.Sc. degrees in Information Technology and
Electrical Engineering from the University of Zagreb, Croatia and the Technical University of Munich (TUM), Germany,
respectively, in 2017 and 2019. Since 2020, he has been a PhD student at the
Chair of Information-oriented Control, TUM School of Computation, Information and Technology. His current research interests include operator-theoretic machine
learning and data-driven control of uncertain systems.
\end{IEEEbiography}
\preBio
\begin{IEEEbiography}[{\bioPic{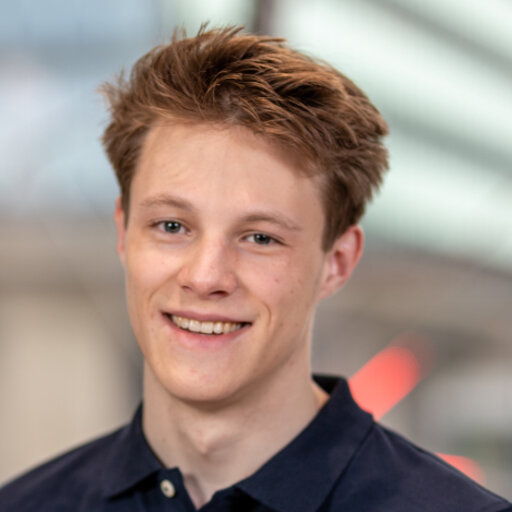}}]{Bas Driessen} is a graduate in Mechanical Engineering from Eindhoven University of Technology (TU/e), where he obtained his master's degree in 2023 with a focus on systems and control. His master's thesis was part of an internship at the Chair of Information-oriented Control at TUM School of Computation, Information and Technology, Technical University of Munich (TUM). His research
interests include modeling and identification of nonlinear systems, machine learning techniques and predictive control.
\end{IEEEbiography}
\preBio
\begin{IEEEbiography}[{\bioPic{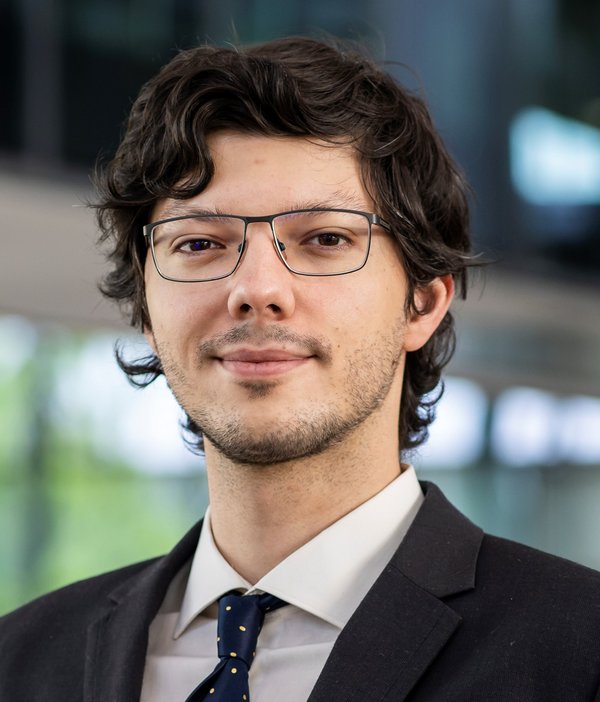}}]{Lucian Cristian Iacob} is a Doctoral Candidate at the Control Systems (CS) Group in the Department of Electrical Engineering. His current research is on modeling and analysis of nonlinear systems using the Koopman and Linear Parameter-Varying (LPV) frameworks, under the supervision of Roland Tóth and Maarten Schoukens, in the Automated Linear Parameter-Varying Modeling and Control Synthesis for Nonlinear Complex Systems (ARPOCS) project. His main research interests include modeling and identification of nonlinear systems and machine learning techniques.
\end{IEEEbiography}
\preBio
\begin{IEEEbiography}[{\bioPic{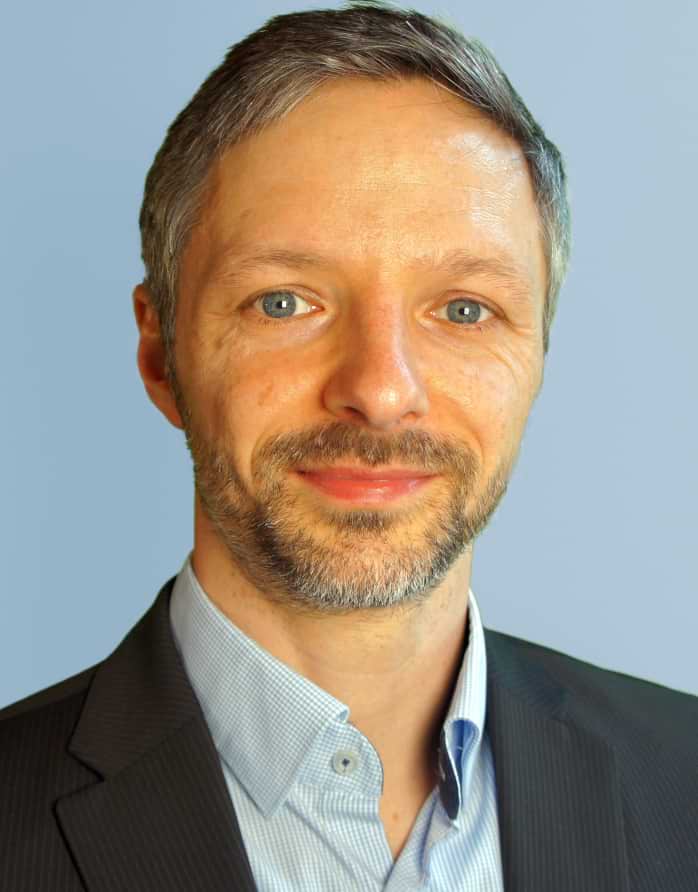}}]{Stefan Sosnowski}
received the Dipl.-Ing. degree and
Dr.-Ing. degree in electrical engineering from the
Technical University of Munich (TUM), Munich,
Germany, in 2007 and 2014, respectively.
Since 2014, he is a Postdoctoral Fellow with the Chair of
Information-oriented Control, at the School of Computation, Information and Technology at TUM. His research interests include bioinspired and underwater robotics, nonlinear control, distributed dynamical systems, and multi-agent systems.
\end{IEEEbiography}
\preBio
\begin{IEEEbiography}[{\bioPic{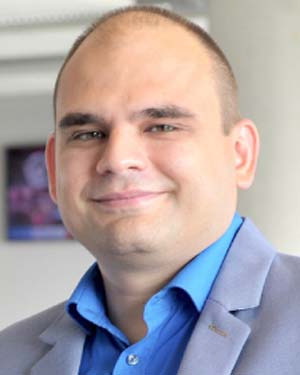}}]{Roland Tóth} received the Ph.D. degree in electrical engineering with Cum Laude distinction from the Delft Center for
Systems and Control (DCSC), Delft University
of Technology (TU Delft), Delft, The Netherlands, in 2008.
He was a Postdoctoral Research Fellow with
TU Delft in 2009 and Berkeley in 2010. He held
a position with DCSC, TU Delft in 2011–2012.
He is currently a Full Professor with the Control
Systems Group, Eindhoven University of Technology, Eindhoven, The Netherlands, and a Senior Researcher with the HUN-REN Institute for Computer Science and Control (SZTAKI), Budapest, Hungary. His research interests include identification and control
of linear parameter-varying (LPV) and nonlinear systems, developing
machine learning methods with performance and stability guarantees, model predictive control and behavioral system
theory with a wide range of applications, including precision mechatronics and autonomous vehicles.
He is a Senior Editor of the IEEE Transactions on Control Systems Technology and an Associate Editor of Automatica. 
\end{IEEEbiography}
\preBio
\begin{IEEEbiography}[{\bioPic{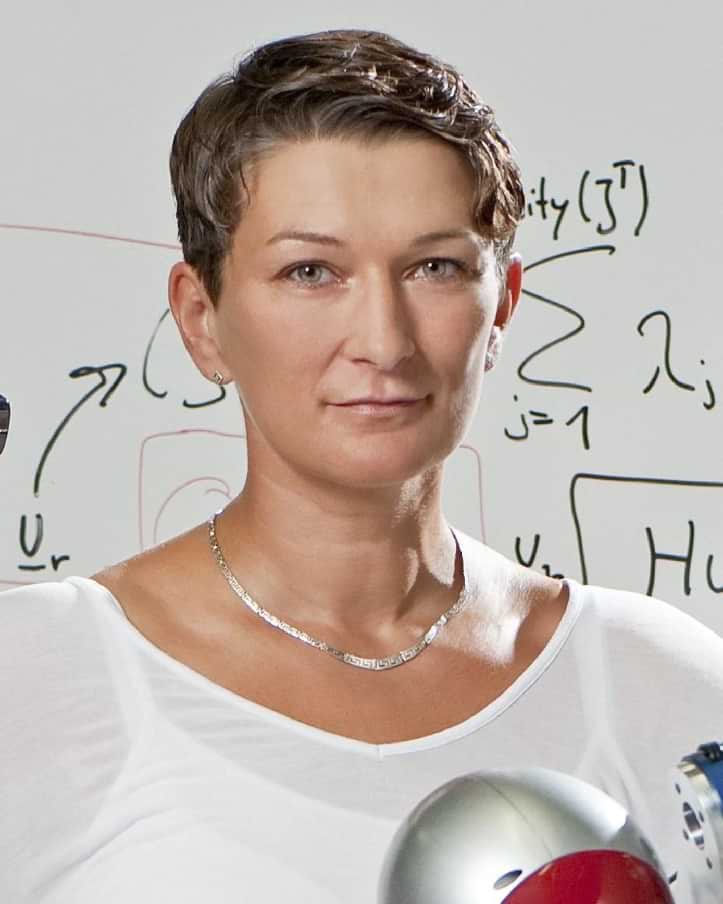}}]{Sandra Hirche} holds the TUM Liesel Beckmann
Distinguished Professorship and heads the Chair
of Information-oriented Control at TUM School of Computation, Information and Technology, Technical University of Munich (TUM), since
2013. She received the diploma engineer degree
in Aeronautical and Aerospace Engineering in 2002
from the Technical University Berlin, Germany, and
the Doctor of Engineering degree in Electrical and
Computer Engineering in 2005 from TUM. Her main research interests include learning, cooperative,
and distributed control with application in human-robot interaction, multirobot systems, and general robotics.
\end{IEEEbiography}
\end{document}